\documentclass{tlp}

\usepackage{graphicx}
\usepackage{ulem}
\usepackage{color}
\usepackage{latexsym}
\usepackage{amsmath}
\usepackage{amsfonts}

\long\def\comment#1{}

\newcommand{\la}{\langle}
\newcommand{\ra}{\rangle}
\newcommand{\rrarrow}{\longrightarrow}

\newtheorem{definition}{Definition} 
\newtheorem{proposition}{Proposition} 
\newtheorem{example}{Example} 
\newtheorem{corollary}{Corollary} 
\newtheorem{theorem}{Theorem} 
\newtheorem{lemma}{Lemma}

\title[Unfolding for CHR programs]
      {Unfolding for CHR programs}

\author[M. Gabbrielli, M.C. Meo, P. Tacchella and H. Wiklicky]{
	Maurizio Gabbrielli\\
		Department of Computer Science/Focus \\
		Universit\`a di Bologna/INRIA\\
		Mura Anteo Zamboni 7,\\
		40127 Bologna, Italy\\
		\email{gabbri@cs.unibo.it}
\and
        Maria Chiara Meo\\
		Dipartimento di Economia\\
		Universit\`a di Chieti-Pescara\\
		Viale Pindaro 42,\\
		65127 Pescara, Italy\\
		\email{cmeo@unich.it}
\and
        Paolo Tacchella\\
		Universit\`a di Bologna\\
		Mura Anteo Zamboni 7,\\
		40127 Bologna, Italy\\
		\email{Paolo.Tacchella@cs.unibo.it}
		\and
        	Herbert Wiklicky\\
		Imperial College London\\
		180 Queen's Gate\\
		London SW7 2BZ, UK\\
		\email{herbert@doc.ic.ac.uk}
}

\submitted{25 October 2010}
\revised{12 September 2012}
\accepted{28 June 2013}

\label{firstpage}

\begin{document}

\bibliographystyle{acmtrans}

\maketitle

\begin{abstract}
Program transformation is an appealing technique which allows to
improve  run-time efficiency, space-consumption, and more generally
to optimize a given program. Essentially, it consists of a sequence
of syntactic program manipulations which preserves some kind of
semantic equivalence. Unfolding is one of the basic operations which is used
by most program transformation systems and which
consists in the replacement of a procedure call by its definition.
While there is a large body of literature on transformation and
unfolding of sequential programs, very few papers have addressed
this issue for concurrent languages.

This paper defines an unfolding system
for CHR programs. We define an unfolding rule, show its
correctness and discuss some conditions which can be used to
delete an unfolded rule while preserving the program meaning.
We also prove that, under some suitable conditions, confluence and termination are preserved by the above transformation.

{\em To appear
in Theory and Practice of Logic Programming (TPLP).}

\end{abstract}

\section{Introduction}

Constraint  Handling Rules (CHR) \cite{Fru98,FA03,Fru08} is a
concurrent, committed-choice language which was initially designed for writing
constraint solvers and which is nowadays a general purpose language. A CHR program is a (finite) set of guarded rules, which
allow to transform multisets of atomic formulas (constraints) into simpler ones.

There exists a very  large body of literature on
CHR, ranging from theoretical aspects to implementations and
applications. However,
only few papers, notably  \cite{FH03,Fru04,SSD05b,TMG07,Tac08,SS09}, consider  source to source transformation of CHR programs. This is not surprising,
since program transformation is in general very difficult for
(logic) concurrent languages and in case of CHR it is even more
complicated, as we discuss later.  Nevertheless, the study of this technique
for concurrent languages and for CHR in particular, is important
as it could lead to significant improvements in the run-time efficiency and space-consumption of programs.

Essentially, a source to source transformation consists of a sequence
of syntactic program manipulations which preserves some kind of semantics. A basic manipulation is {\em unfolding}, which consists in the replacement
of a procedure call by its definition. While this operation can be
performed rather easily for sequential languages, and indeed  in
the field of logic programming it was first investigated by Tamaki
and Sato more than twenty years ago \cite{TS84}, when considering
logic concurrent languages it becomes quite difficult to define
reasonable conditions which ensure its correctness.

In this paper, we first define an unfolding rule for CHR programs
and show that it preserves the semantics of the program in terms
of qualified answers \comment{, a notion already defined in the literature}
\cite{Fru98}.
Next, we provide a syntactic condition which allows one to replace in a
program a rule by its unfolded version while preserving
qualified answers.  This condition preserves also termination, provided that one considers normal derivations.
We also show that a more restricted condition ensures
that confluence is preserved.
Finally, we give a weaker condition for replacing a rule by its unfolded version:
This condition allows to preserve the qualified answers for a specific class of programs (those which are normally terminating and confluent).

Even though the idea of the unfolding is
straightforward, its technical development is complicated by the
presence of guards, multiple heads, and matching substitution,
as previously mentioned. In
particular, it is not obvious to identify conditions which
allow to replace the original rule by its unfolded version.
Moreover, a further reason of complication comes from the fact
that we consider as reference semantics the one
defined in \cite{DSGH04} and called $\omega_t$, which avoids trivial non-termination by
using a  ``token store'' (or history). The token store
idea was originally introduced by \cite{Abd97} but the shape of these
tokens is different from that of those used  in \cite{DSGH04}. Due
to the presence of this token store,  in order to define correctly
the unfolding we have to slightly modify the syntax of CHR programs
by adding to each rule  a local token store. The resulting programs
are called annotated and we define their semantics by providing a
(slightly) modified version of the semantics $\omega_t$, which is
proven to preserve the qualified answers.

The remainder of this paper is organized as follows. Section~\ref{sec:notation}
introduces the CHR syntax while the operational semantics $\omega_t$
\cite{DSGH04} and the modified one $\omega'_t$  are given in
Section~\ref{sec:semantics}. Section~\ref{sec:unfolding} defines
the unfolding rule (without replacement) and proves its correctness.
Section~\ref{sec:safty-rule-deletion}  discuss the problems
related to the replacement of a rule by its unfolded version and provides a correctness condition for such a replacement.
In this section, we also prove that (normal) termination and
confluence are preserved by the replacement which satisfies this condition. A further, weaker, condition ensuring the correctness of replacement for (normally) terminating and confluent programs is given in Section~\ref{sec:normally -rep}. Finally, Section~\ref{sec:conclusion_and_future} concludes
by discussing some related works. Some of the proofs are deferred to the Appendix in order to improve the readability of the paper.

A preliminary version of this paper appeared in \cite{TMG07}, some results were contained in the thesis  \cite{Tac08}.

\section{Preliminaries}\label{sec:notation}

In this section, we introduce the syntax of CHR and some notations
and definitions we will need in the paper. For our purpose, a {\em constraint} is simply defined as an atom
$p(t_1, \ldots, t_n)$, where $p$ is some predicate symbol of arity $n \geq 0$ and
$(t_1, \ldots, t_n)$ is an $n$-tuple of terms.
A {\em term} is (inductively) defined as a variable $X$, or as $f(t_1, \ldots, t_n)$, where $f$ is a function symbol of arity $n \geq 0$ and
$t_1, \ldots, t_n$ are terms. ${\cal T}$ is the set of all terms.

We use the following notation: let $A$ be any syntactic object and let $V$ be a set of variables.
$\exists_{V} A$ denotes the
existential closure of $A$ w.r.t. the variables in $V$, while $\exists_{-V} A$ denotes the existential closure of $A$ with the exception of the variables in $V$ which remain unquantified. $Fv(A)$ denotes the free variables appearing in $A$.

We use ``,'' rather than $\wedge$ to denote conjunction
and we will often consider a conjunction of atomic constraints as
a multiset of atomic constraints.
We use $\mathrel+\joinrel\mathrel+$ for sequence concatenation, $\epsilon$ for empty sequence, $\setminus$ for set difference operator and $\uplus$ for multiset union. We shall sometimes treat multisets as sequences (or vice versa), in which case we nondeterministically choose an order for the objects in the multiset. We use the notation
$p(s_1, \ldots, s_n) = p(t_1, \ldots, t_n)$ as a shorthand for the (conjunction of) constraints
$\,s_1=t_1,  \ldots, s_n=t_n$. Similarly if $S\equiv s_1, \ldots, s_n$ and $T\equiv t_1, \ldots, t_n$ are sequences of equal length then $S=T$ is a shorthand for $\,s_1=t_1, \ldots, s_n=t_n$.

A substitution is a mapping $\vartheta: V
\rightarrow {\cal T}$ such that the set $dom(\vartheta) = \{X
\mid \vartheta(X) \neq X\}$ (domain of $\vartheta$) is finite;
$\varepsilon$ is the empty substitution: $dom(\varepsilon)
=\emptyset$.
\comment{If $\vartheta$ is a substitution and $E$ is a syntactic expression we denote by
$\vartheta_{\mid E}$ the restriction of $\vartheta$
to the variables in $Fv(E)$.}

The composition $\vartheta\sigma$ of the substitutions $\vartheta$
and $\sigma$ is defined as the functional composition. A substitution
$\vartheta$ is idempotent if
$\vartheta \vartheta = \vartheta$.
A renaming is a (nonidempotent) substitution $\rho$ for which there exists the inverse $\rho^{-1}$
such that $\rho\rho^{-1}$ = $\rho^{-1}\rho$ =$\varepsilon$.
\comment{The preordering $\preceq$ (more general than) on substitutions is
such that $\vartheta\preceq\sigma$ iff there exists a $\gamma$
such that $\vartheta\gamma$ = $\sigma$.
The result of the
application of the substitution $\vartheta$ to an expression $E$ is an
{\em instance} of $E$ denoted by $E\vartheta$. We define $E\preceq
E'$ ($E$ is more general than $E'$) iff there exists a $\vartheta$
such that $E\vartheta$ = $E'$.}
\comment{A substitution $\vartheta$ is a
grounding for $t$ if $t\vartheta$ is ground and $Ground(t)$ denotes the set of ground instances of $t$.}

We restrict our attention to idempotent substitutions, unless explicitly stated otherwise.

Constraints can be divided into either {\em user-defined} (or CHR) constraints or {\em built-in} constraints on some constraint domain ${\cal D}$. The built-in constraints are handled by an existing solver and we assume
given a (first order) theory $\mathcal{CT} $ which describes their meaning.
We assume also that the built-in constraints contain the predicate $=$ which is described,
as usual, by the Clark Equality Theory \cite{Llo84} and the values
\texttt{true} and \texttt{false} with their obvious meaning.

We use $c,d$ to denote built-in constraints, $h,k,f,s, p, q$ to denote CHR
constraints, and $a,b, g$ to denote both built-in and user-defined
constraints (we will call these generically constraints). The capital
versions will be used to denote multisets
(or sequences) of constraints.

\subsection{CHR syntax}\label{sec:syntax}
As shown by the following definition \cite{Fru98}, a  \emph{CHR program}
consists of a set of rules which can be divided into three types:
\emph{simplification}, \emph{propagation}, and \emph{simpagation}
rules. The first kind of rules is used to rewrite
CHR constraints into simpler ones, while second kind allows to add new redundant
constraints which may cause further simplification. Simpagation rules allow to represent both simplification and propagation rules.

\begin{definition}[{\sc CHR Syntax}]
A CHR program is a finite set of CHR rules.
There are three kinds of CHR rules:

\noindent{A \textbf{simplification} rule has the form: $${\it r}@ H
\Leftrightarrow C  \,|\, B$$}\\
\noindent{A \textbf{propagation} rule has the form:
$${\it r}@ H \Rightarrow C  \,|\,  B$$}\\
\noindent{A \textbf{simpagation} rule has the form:
$${\it r}@ H_1
\setminus H_2 \Leftrightarrow C  \,|\,  B,$$}
\noindent where ${\it r }$ is a unique identifier of a rule, $H$,
$H_1$ and $H_2$ are sequences of user-defined constraints, with $H$  and
$H_1 \mathrel+\joinrel\mathrel+ H_2$ different from the empty sequence, $C$ is a possibly empty conjunction of built-in constraints,
and $B$ is a possibly empty sequence of (built-in and
user-defined) constraints. $H$ (or $H_1 \setminus H_2 $) is called
\emph{head}, $C$ is called \emph{guard} and $B$ is called \emph{body} of the rule.
\end{definition}

A  \emph{simpagation} rule can simulate both simplification and
propagation rule by considering, respectively, either $H_1$ or
$H_2$ empty. In the following, we
will consider in the formal treatment only simpagation rules.

\subsection{CHR Annotated syntax}\label{sec:annsyntax}

When considering unfolding we need to consider a slightly
different syntax,  where rule identifiers are not necessarily
unique, each atom in the body is associated with an identifier,
that is unique in the rule, and
where each rule is associated  with a local token store $T$.
More precisely, we define an identified CHR constraint  (or
identified atom) $h\#i$ as a CHR constraint $h$ associated with
an integer $i$ which allows to distinguish different copies of the
same constraint.

Moreover, let us define a token as an object of the form $r@i_1, \ldots, i_l$, where $r$
is the name of a rule and $i_1, \ldots, i_l$ is a sequence of distinct identifiers.
\comment{(namely, the sequence of identifiers associated with the constraints to which the
head of the rule $r$ is applied).}
A token store (or history) is a set of tokens.

\begin{definition}[{\sc CHR Annotated syntax}]
An {\em annotated} rule has then the form:
$${\it r}@ H_1
\setminus H_2 \Leftrightarrow C  \,|\,B;T$$ where ${\it r }$
is an identifier, $H_1$ and $H_2$ are
sequences of user-defined constraints with
$H_1 \mathrel+\joinrel\mathrel+ H_2$ different from the empty sequence, $C$ is a possibly empty conjunction of built-in constraints,
$B$ is a possibly empty sequence of built-in and identified CHR constraints
such that different (occurrences of) CHR constraints have different identifiers,
and $T$ is a token store.
$H_1 \setminus H_2 $ is called
\emph{head}, $C$ is called \emph{guard}, $B$ is called \emph{body} and
$T$ is called \emph{local token store} of the annotated rule.
An annotated CHR program is a finite set of annotated CHR rules.

\end{definition}

We will also use the following two functions: {\em chr(h$\#$i)=$_{def}$ h} and the overloaded
function {\em id(h$\#$i)=$_{def}$ i}, (and $id(r@i_1,\ldots,i_l)=_{def}\{i_1, \ldots, i_l\}$),
extended to sets and sequences of identified CHR constraints (or tokens) in
the obvious way. An (identified) CHR {\em goal} is a multi-set of
both (identified) user-defined and built-in constraints.
$Goals$ is the set of all (possibly
identified) goals.

Intuitively, identifiers are used to distinguish different
occurrences of the same atom in a rule or in a goal. The identified atoms can
be obtained by using a suitable function which associates a
(unique) integer to each atom. More precisely, let $B$ be a goal
which contains $m$ CHR-constraints. We assume that the function
$I(B)$ identifies each CHR constraint in $B$ by
associating to it a unique integer in $[1,m]$ according to the
lexicographic order.

The token store allows one to memorize  some tokens,
where each token describes which propagation rule has been used
for reducing which identified atoms.
As we discuss
in the next section, the use of this information was originally proposed
in \cite{Abd97} and then further elaborated in the semantics defined
in \cite{DSGH04} in order to avoid trivial  non-termination
arising from the repeated application of the same propagation rule
to the same constraints. Here, we simply incorporate this
information in the syntax, since we will need to manipulate it in
our unfolding rule.

Given a CHR program $P$, by using the function $I(B)$ and
an initially empty local token store we can construct its annotated
version as the next definition explains.

\begin{definition}
Let $P$ be a CHR program. Then its annotated version is defined as follows:
\[\begin{array}{lll}
  Ann(P)=\{&{\it r}@ H_1 \setminus H_2 \Leftrightarrow C  \,|\,  I(B);\emptyset \hbox{ such that }& \\
& {\it r}@ H_1 \setminus H_2 \Leftrightarrow C
 \,|\,
B \in P &\}.
\end{array}
\]
\end{definition}

\noindent {\bf Notation} \\
In the following examples, given a
(possibly annotated) rule
$${\it r}@ H_1 \setminus
H_2 \Leftrightarrow C  \,|\,  B(;T),$$   we write it as
$${\it r}@ H_2 \Leftrightarrow C  \,|\,  B(;T),$$ if $H_1$ is empty and
we write it as
$${\it r}@ H_1 \Rightarrow C  \,|\,  B(;T),$$ if $H_2$ is empty.
That is, we maintain also the notation previously introduced for
simplification and propagation rules. Moreover, if $C=
\texttt{true}$, then $\texttt{true}  \,|\, $ is omitted
and if in an annotated rule the token store is empty we
simply  omit it.
Sometimes, in order to simplify the notation, if in an annotated program $P$ there are no annotated propagation rules, then we write $P$ by using the standard syntax.\\

Finally, we will use $cl, \
cl', \ldots $ to denote (possibly annotated) rules
and
$cl_r, \ cl'_r, \ldots $ to denote (possibly annotated) rules with identifier $r$.

\begin{example}\label{chiara}
The following CHR program, given a forest of finite trees
(defined in terms of the predicates $\text{root}$ and $\text{edge}$, with the obvious meaning), is
able to recognize if two nodes belong to the same tree and if so returns the root.

The program $P$ consists of the following five rules
$$
\begin{array}{l}
r_1@\text{root} (V), \text{same} (X, Y ) \Rightarrow X=Y,  X=V \mid \text{success}(V)\\
r_2@\text{root} (V), \text{same} (X, Y ) \Leftrightarrow  X\neq Y \mid \text{root} (V), \text{same} (V,X), \text{path}(V,Y)\\
r_3@ \text{path}(I,J) \Rightarrow I=J \mid \texttt{true}\\
r_4@ \text{edge} (U,Z) \setminus \text{path}(I,J) \Leftrightarrow  J=Z \mid \text{path}(I,U)\\
r_5@\text{root} (V) \setminus \text{path}(I,J) \Leftrightarrow  V=J, V\neq I \mid \texttt{false}
\end{array}
$$

Then its annotated version $Ann(P)$ is defined as follows:
$$
\begin{array}{l}
r_1@\text{root} (V), \text{same} (X, Y ) \Rightarrow X=Y,  X=V \mid \text{success}(V)\#1;\emptyset\\
r_2@\text{root} (V), \text{same} (X, Y ) \Leftrightarrow X\neq Y \mid \text{root} (V)\#1, \text{same} (V,X)\#2, \text{path}(V,Y)\#3;\emptyset\\
r_3@ \text{path}(I,J) \Rightarrow I=J \mid \texttt{true};\emptyset\\
r_4@ \text{edge} (U,Z) \setminus \text{path}(I,J) \Leftrightarrow  J=Z \mid \text{path}(I,U)\#1;\emptyset\\
r_5@\text{root} (V) \setminus \text{path}(I,J) \Leftrightarrow  V=J, V\neq I \mid \texttt{false};\emptyset
\end{array}
$$

\comment{$$
\begin{array}{l}
r_1@\text{isin} (X, Y ) \Rightarrow X=Y \mid \text{success}\#1;\emptyset\\
r_2@\text{isin} (X, Y ) \Leftrightarrow  X \neq Y \mid \text{path}(V,X)\#1,\text{isin} (V,Y)\#2;\emptyset\\
r_3@ \text{edge} (V,X) , \text{path}(Y,Z) \Rightarrow X=Z \mid V=Y;\emptyset\\
r_4@ \text{edge} (V,X) , \text{path}(X,Z) \Leftrightarrow  \text{edge} (V,X)\#1,\text{path}(V,Z)\#2;\emptyset\\
r_5@\text{path} (V,Z ), \text{isin}(X,Y)\Leftrightarrow  X=Z \mid \text{isin} (V,Y)\#1;\emptyset
\end{array}
$$
$$
\begin{array}{l}
r_1@\text{isin} (X, Y ) \Rightarrow X=Y \mid \text{success}\\
r_2@\text{path}(V,Z) \setminus \text{isin} (X, Y ) \Leftrightarrow  X \neq Y, Z=X \mid
\text{isin} (V,Y)\\
r_3@\text{path}(V,Z) , \text{isin} (X, Y ) \Leftrightarrow  X \neq Y, Y=V \mid \text{path}(V,Z),
\text{isin} (X,Z)\\
r_4@ \text{edge} (V,X)  \Rightarrow  \text{path}(V,X)\\
r_5@ \text{edge} (V,X) , \text{path}(X,Z) \Leftrightarrow  \text{edge} (V,X),\text{path}(V,Z)
\end{array}
$$
Then its annotated version $Ann(P)$ is defined as follows:
$$
\begin{array}{l}
r_1@\text{isin} (X, Y ) \Rightarrow X=Y \mid \text{success}\#1;\emptyset\\
r_2@\text{path}(V,Z) \setminus \text{isin} (X, Y ) \Leftrightarrow  X \neq Y, Z=X \mid
\text{isin} (V,Y)\#1;\emptyset\\
r_3@\text{path}(V,Z) , \text{isin} (X, Y ) \Leftrightarrow  X \neq Y, Y=V \mid \text{path}(V,Z)\#1,
\text{isin} (X,Z)\#2;\emptyset\\
r_4@ \text{edge} (V,X)  \Rightarrow  \text{path}(V,X)\#1;\emptyset\\
r_5@ \text{edge} (V,X) , \text{path}(X,Z) \Leftrightarrow  \text{edge} (V,X)\#1,\text{path}(V,Z)\#2;\emptyset\\
\end{array}
$$}
\end{example}

\section{CHR operational semantics}\label{sec:semantics}

This section first introduces the reference semantics $\omega_t$ \cite{DSGH04}. For the sake of simplicity, we omit indexing the relation with the name of the program.

Next, we define a slightly different operational  semantics,
called $\omega_t'$, which considers annotated programs and which
will be used to prove the correctness of our unfolding rules (via some form of
equivalence between  $\omega_t'$ and $\omega_t$).

In the following, given a (possibly annotated) rule
$cl_r={\it r}@ H_1 \setminus
H_2 \Leftrightarrow C  \,|\,  B(;T),$
we denote by $\exists_{cl_r}$ the existential quantification
$\exists _{Fv(H_1,H_2,C,B)}$. By an abuse of notation, when it is clear from the context, we will write
$\exists_{r}$ instead of $\exists_{cl_r}$.

\subsection{The semantics $\omega_t$}

\begin{table*}[tbp]
\caption{The transition system $T_{\omega_t}$ for the $\omega_t$ semantics}
\centering
\label{omega-t}

\begin{tabular}{lll}
\hline\noalign{\smallskip}
&\mbox{   }&\mbox{   }
\\
\mbox{\bf Solve} &  $\displaystyle{\frac {c \hbox{
is a built-in constraint}} {\la \{c\}\uplus G,   S,C, T \ra_n
\rrarrow_{\omega_t} \la G,   S, C\wedge c, T \ra_n}} $

&\mbox{ }
\\
&\mbox{   }&\mbox{   }
\\

\mbox{\bf Introduce}& $\displaystyle{\frac {h \hbox{
is a user-defined constraint}}{ \la \{h \}\uplus G,   S,C, T
\ra_n \rrarrow_{\omega_t} \la G,\{ h\#n \} \uplus   S ,C, T
\ra_{n+1} }}$ &\mbox{ }
\\
&\mbox{   }&\mbox{   }
\\

\mbox{\bf Apply}& $\displaystyle{\frac {
\begin{array}{c}
{\it r}@H'_1 \setminus H'_2 \Leftrightarrow D  \,|\,  B \in P \ \ \ \\
\mathcal{CT}  \models C \rightarrow \exists _{r}
((chr(  H_1,   H_2)=(H'_1,  H'_2))\wedge D)
\end{array}
}
{\displaystyle
\begin{array}{c}
\la G, H_1 \uplus H_2 \uplus
  S,C, T \ra_n \rrarrow_{\omega_t}\\ \la B\uplus G,
H_1 \uplus   S, (chr(  H_1,   H_2)=(H'_1,
H'_2))\wedge D \wedge  C, T' \ra_n
\end{array} }}$&\mbox{ }
\\
&\mbox{   }&\mbox{   }
\\
&\mbox{where }
  ${\it r}@id (H_1,H_2) \not \in T $
\mbox{ and } $T'=T \cup \{{\it r}@id(H_1,H_2)\} $&\mbox{
}
\\
&\mbox{   }&\mbox{   }
\\
\noalign{\smallskip}\hline
\end{tabular}
\end{table*}

We describe the operational semantics $\omega_t$, introduced in {\cite{DSGH04}, by using a transition system
\[T_{\omega_t}= ({\it Conf_t},
\rrarrow_{\omega_t}).\] Configurations in ${\it Conf_t}$ are
tuples of the form $\langle G,   S,C, T\rangle_n$ where $G$, the {\em goal store} is a multiset of constraints. The \emph{CHR constraint store} $S$ is a set of
identified CHR constraints. The \emph{built-in constraint store} $C$ is a
conjunction of built-in constraints. The {\em propagation history}
$T$ is a token store and $n$ is an integer.
Throughout this paper, we use the symbols $\sigma, \sigma', \sigma_i, \ldots$ to represent configurations in ${\it Conf_t}$.

The \emph{goal store} ($G$) contains all constraints to
be executed. The \emph{CHR constraint store} ($S$) is the set\footnote{Note that sometimes we treat $S$ as a multiset. This is the case, for example, of the transition rules, where considering $S$ as a multiset simplifies the notation.} of
identified CHR constraints that can be matched with the head of the
rules in the program $P$. The \emph{built-in constraint store} ($C$) contains any built-in  constraint that has been passed to the built-in constraint solver. Since we will usually have no information about the
internal representation of $C$, we treat it as a conjunction of constraints. The {\em propagation history}
($T$) describes which rule has been used
for reducing which identified atoms. Finally, the \emph{counter} $n$ represents the next
free integer which can be used to number a CHR constraint.

\comment{We describe the operational semantics $\omega_t$, introduced in \cite{DSGH04}, by using a transition system
\[T_{\omega_t}= ({\it Conf_t},
\rrarrow_{\omega_t}).\] Configurations in ${\it Conf_t}$ are
tuples of the form $\langle G,   S,C, T\rangle_n$ with the
following meaning. The \emph{goal store} $G$ is a multiset (repeats are allowed) of constraints to
be evaluated. The \emph{CHR constraint store} $S$ is the set of
identified CHR constraints that can be matched with the head of the
rules in the program $P$. The \emph{built-in constraint store} $C$ is a
conjunction of built-in constraints. The {\em propagation history}
$T$ is a token store. Finally, the \emph{counter} $n$ represents the next
free integer which can be used to number a CHR constraint.}

Given a goal $G$, the  {\em initial configuration} has the form
\[\langle G,\emptyset,\texttt{true}, \emptyset \rangle_1.\]
A {\em final configuration} has either the form $\langle  G',
  S, {\tt false} , T\rangle_n$, when it is {\em failed}, or it
has the form $\langle \emptyset,  S, C,T \rangle_n$ (with $\mathcal{CT}  \models C \not \leftrightarrow {\tt false}$) when it
represents a successful termination (since there are no more
applicable rules).

The relation $\rrarrow_{\omega_t}$ (of the transition system
$T_{\omega_t}$) is defined  by the rules
in Table~\ref{omega-t}: the \textbf{Solve} rule moves a
built-in constraint from the goal store to the built-in constraint
store; the \textbf{Introduce} rule identifies and moves a CHR (or user-defined) constraint from the goal store to the CHR constraint
store;  the \textbf{Apply} rule chooses a program rule $cl$ and fires it, provided that the following conditions are satisfied: there exists a matching between the constraints in the CHR store and the ones in
the head of $cl$; the guard of $cl$ is
entailed by the built-in constraint store (taking into account also the
matching mentioned before); the token that would be added by \textbf{Apply}}
to the token store is not already
present. After the application of $cl$,
the constraints which match with the right hand side of the head
of $cl$ are deleted from the CHR constraint store, the body of $cl$ is added to
the goal store and the guard of $cl$, together with the equality representing the matching, is added to the built-in constraint store.
The \textbf{Apply} rule assumes that all
the variables appearing in a program clause are renamed with fresh ones in order to
avoid variable names clashes.

From the rules, it is clear that when not considering
tokens (as in the original semantics of \cite{Fru98}) if a propagation rule can be
applied once  then it can be applied infinitely many times, thus
producing an infinite computation (no fairness assumptions are
made here). Such a trivial non-termination is avoided by tokens, since they ensure that
if a propagation rule is used to reduce a sequence of constraints
then the same rule has not been used before on the same
sequence of constraints.

\subsection{The modified semantics $\omega_t'$}
We now define the semantics $\omega_t'$ which considers annotated rules.
This semantics differs from $\omega_t$ in two aspects.

First, in $\omega_t'$
the goal store and the CHR store are fused in a unique generic \emph{store}, where
CHR constraints are immediately labeled. As a consequence, we do not
need  the {\bf Introduce} rule anymore and every CHR
constraint in the body of an applied rule is immediately
utilizable for rewriting.

The second difference concerns the shape of the rules. In fact,
each annotated rule $cl$ has  a local token store (which can be
empty) that is associated with it and which is used to keep track of
the propagation rules that are used to unfold the body of $cl$.
Note also that here, differently from the case of the propagation
history in $\omega_t$, the token store associated with a computation can be updated
by adding multiple tokens at once (because an unfolded rule with many
tokens in its local token store has been used).

In order to define  $\omega_t'$ formally,
we need a function $inst$  which updates
the formal identifiers of a rule to the actual computation ones. Such a function 
is defined as follows.

\begin{definition}\label{definst}
 Let $Token$ be the set of all possible token
sets and let $\mathbb{N}$ be the set of natural numbers. We denote
by $inst: Goals \times Token\times \mathbb{N}  \rightarrow
Goals \times Token\times \mathbb{N}$ the function such that
$inst(  B,T,n)=(B', T',m$), where
\begin{itemize}
    \item $B$ is an identified CHR goal,
    \item $(B', T') $ is obtained from $(B, T)$ by
    incrementing each identifier in $(B, T)$ with $n$ and
    \item $m$ is the greatest identifier in $(B', T')$.
\end{itemize}
\end{definition}

We describe now the operational semantics $\omega_t'$  for annotated CHR
programs by using, as usual, a transition system
\[T_{\omega'_t}= ({\it Conf'_{t}}, \rrarrow_{\omega'_t}).\] Configurations in
${\it Conf'_t}$ are tuples of the form $\langle S,C,
T\rangle_n$ with the following meaning.  $S$ is the set\footnote{Also in this case, sometimes we treat $S$ as a multiset. See the previous footnote.} of
identified CHR constraints that can be matched with rules in the
program $P$ and built-in constraints. The built-in constraint
store $C$ is a conjunction of built-in constraints and $T$ is a
set of tokens, while the counter $n$ represents the last integer
which was used to number the CHR constraints in $S$.

Given a goal $G$, the  {\em initial configuration} has the form
\[\langle I(G),\texttt{true}, \emptyset \rangle_m,\]
where $m$ is the number of CHR constraints in $G$ and $I$ is the function which associates the identifiers with the CHR constraints in $G$.
A {\em failed
configuration} has the form $\langle   S, {\tt false}, T\rangle_n$.

A {\em final
configuration} either is failed or it has the form $\langle
  S, C,T \rangle_n$ (with $\mathcal{CT}  \models C \not \leftrightarrow {\tt false}$) when it represents a successful
termination, since there are no more applicable rules.

The relation $\rrarrow_{\omega'_t}$ (of the transition system $T_{\omega'_t}$) is defined  by the rules in
Table~\ref{tab:operational-semantics} which have the following explanation:
\begin{description}
\item[\textbf{Solve'}]{moves a built-in constraint from the store to the
built-in constraint store;}

\item[\textbf{Apply'}] {fires  a rule $cl$ of the form $r@H_1'\backslash H_2'\Leftrightarrow
D \, |\,   B; T_r$ provided that the following conditions are satisfied: there exists a matching between the constraints in the store and the ones in
the head of $cl$; the guard of $cl$ is
entailed by the built-in constraint store (taking into account also the
matching mentioned before); $r@id( H_1, H_2)\not \in
T$. These conditions are equal to those already seen for {\bf Apply}. Moreover, analogously to the {\bf Apply} transition step,
$chr (H_1,  H_2)= (H_1',H_2')$ together with $D$ are added to the built-in constraint
store.
However, in this case, when the rule $cl$ is fired, $  H_2$ is replaced by $ B$ and the local store $T_r$ is added to $T$ (with $r @id(H_1,H_2)$), where each identifier
is suitably incremented by the $inst$ function.  Finally, the subscript $n$ is replaced by $m$, that is
the greatest number used during the computation step.

As for the \textbf{Apply} rule, the  \textbf{Apply'} rule assumes that all
the variables appearing in a program clause are renamed with fresh ones in order to
avoid variable names clashes.

}
\end{description}

\begin{table*}[t]
\caption{The transition system $T_{\omega'_t}$ for the $\omega'_t$ semantics}
\centering
\label{tab:operational-semantics}
$$
\begin{array}{llll}

\hline\noalign{\smallskip}
&&&\\

\mbox{   }&\textbf{Solve'}&\displaystyle\frac{ c \mbox{ is a built-in
constraint}} {\langle\{c\}\uplus   G, C,
T\rangle_n\rrarrow_{\omega'_t}
\langle   G, c\wedge C, T\rangle_n}&\mbox{   }\\

&&&\\

&\textbf{Apply'}&\displaystyle\frac{
\begin{array}{c}
r@H'_1\backslash H'_2
\Leftrightarrow D\, | \,   B ; T_r\in P,  \quad \quad \\
\mathcal{CT} \models C\rightarrow \exists_r((chr(  H_1,
  H_2)=(H'_1, H'_2))\wedge D)
\end{array} } {
\begin{array}{c}
\langle   H_1\uplus
H_2\uplus   G, C, T \rangle_n\rrarrow _{\omega'_t} \\
\langle   B' \uplus   H_1\uplus   G, (chr(  H_1,
  H_2)=(H'_1, H'_2))\wedge D \wedge C, T'\rangle_{m}
\end{array}
}&\mbox{   }\\

&&&\\

&&\mbox{where  }  (B', T'_r,m)= inst(B,T_r,n),\,
r @id(H_1,H_2)\not\in T \mbox{ and } &\mbox{   }\\
&&T'=T\cup \{r @id(H_1,H_2)\} \cup T'_r
.&\mbox{   }\\
&&&\\
\noalign{\smallskip}\hline
\end{array}
$$
\end{table*}

The following example shows a derivation obtained by the new transition
system.

\begin{example}\label{chiara1}
Given the goal $\text{root} (a), \text{same} (b,c), \text{edge} (a,b), \text{edge}(a,d), \text{edge}(d,c)$ in the following program $P'$,
\begin{small}
$$
\begin{array}{l}
r_1@\text{root} (V), \text{same} (X, Y ) \Rightarrow X=Y,  X=V \mid \text{success}(V)\#1;\emptyset\\
r_2@\text{root} (V), \text{same} (X, Y ) \Leftrightarrow  X\neq Y \mid \text{root} (V)\#1, \text{same} (V,X)\#2, \text{path}(V,Y)\#3;\emptyset\\
r_2@\text{root} (V), \text{same} (X, Y ) \Leftrightarrow  X\neq Y , V=X \mid
\hspace*{-0.15cm}\begin{array}[t]{ll}
\text{root} (V)\#1, \text{same} (V,X)\#2, \text{path}(V,Y)\#3,\\
\text{success}(I)\#4, V=I, V=J, X=L ;\{r_1@1,2\}
\end{array}\\
r_2@\text{root} (V), \text{same} (X, Y ) \Leftrightarrow  X\neq Y , V\neq X \mid  \hspace*{-0.15cm}\begin{array}[t]{ll} \text{path}(V,Y)\#3,
\text{root} (I)\#4, \text{same} (J,L)\#5, \\
I=V, J=V, L=X;\emptyset
\end{array}
\\
r_2@\text{root} (V), \text{same} (X, Y ) \Leftrightarrow  X\neq Y, V=Y \mid
\hspace*{-0.15cm}\begin{array}[t]{ll}\text{root} (V)\#1, \text{same} (V,X)\#2, \\
\text{path}(V,Y)\#3,
V=I, Y=J;\{r_3@3\}
\end{array}
\\
r_3@ \text{path}(I,J) \Rightarrow I=J \mid \texttt{true};\emptyset\\
r_4@ \text{edge} (U,Z) \setminus \text{path}(I,J) \Leftrightarrow  J=Z \mid \text{path}(I,U)\#1;\emptyset\\
r_4@ \text{edge} (U,Z) \setminus \text{path}(I,J) \Leftrightarrow  J=Z , I=U\mid
\text{path}(I,U)\#1, I=X, U=Y;\{r_3@1\}\\
r_5@\text{root} (V) \setminus \text{path}(I,J) \Leftrightarrow  V=J, V\neq I \mid \texttt{false};\emptyset
\end{array}
$$
\end{small}
\vspace*{0.2cm}
we obtain the following derivation \comment{, where we simplify the built-in store,}\\
\begin{small}
$
\begin{array}{l}\langle
({\bf root (a)\#1}, {\bf same (b,c)\#2}, \text{edge} (a,b)\#3, \text{edge}(a,c)\#4, \text{edge}(c,d)\#5), true, \emptyset\rangle_5 \rrarrow _{\omega'_t}
\end{array}
$
\vspace*{0.2cm}\\
$
\begin{array}{l}
\langle
(\text{path}(V_1,Y_1)\#6,
\text{root} (I_1)\#7, \text{same} (J_1,L_1)\#8,
I_1=V_1, J_1=V_1, L_1=X_1, \\
\text{edge} (a,b)\#3, \text{edge}(a,d)\#4, \text{edge}(d,c)\#5), \\
(a=V_1, b=X_1, c=Y_1,X_1\neq Y_1 , V_1\neq X_1), \{r_2@1,2\}\rangle_8 \rrarrow _{\omega'_t}^*
 \end{array}
$
\vspace*{0.2cm}\\
$
\begin{array}{l}
\langle
(\text{path}(V_1,Y_1)\#6,
{\bf root (I_1)\#7}, {\bf\text{same} (J_1,L_1)\#8},
\text{edge} (a,b)\#3, \text{edge}(a,c)\#4, \text{edge}(c,d)\#5), \\
(I_1=V_1, J_1=V_1, L_1=X_1,a=V_1, b=X_1, c=Y_1,X_1\neq Y_1 , V_1\neq X_1), \{r_2@1,2\}\rangle_8 \rrarrow _{\omega'_t}
\end{array}
$
\vspace*{0.2cm}\\
$
\begin{array}{l}
\langle
(\text{root} (V_2)\#9, \text{same} (V_2,X_2)\#10, \text{path}(V_2,Y_2)\#11, \text{success}(I_2)\#12,\\
 V_2=I_2, V_2=J_2, X_2=L_2,
\text{path}(V_1,Y_1)\#6,
\text{edge} (a,b)\#3, \text{edge}(a,c)\#4, \text{edge}(c,d)\#5), \\
(V_2=I_1, X_2=J_1, Y_2=L_1, X_2\neq Y_2, V_2=X _2, I_1=V_1, J_1=V_1, L_1=X_1,\\
a=V_1, b=X_1, c=Y_1,X_1\neq Y_1 , V_1\neq X_1), \{r_2@1,2, \ r_2@7,8, \ r_1@9,10 \}\rangle_{12} \rrarrow _{\omega'_t}^*
\end{array}
$
\vspace*{0.2cm}\\
$
\begin{array}{l}
\langle
(\text{root} (V_2)\#9, \text{same} (V_2,X_2)\#10, {\bf path(V_2,Y_2)\#11}, \text{success}(I_2)\#12,\\
 \text{path}(V_1,Y_1)\#6,
{\bf edge (a,b)\#3}, \text{edge}(a,c)\#4, \text{edge}(c,d)\#5), \\
(V_2=I_2, V_2=J_2, X_2=L_2,V_2=I_1, X_2=J_1, Y_2=L_1, X_2\neq Y_2, V_2=X _2, I_1=V_1, J_1=V_1,\\
 L_1=X_1, a=V_1, b=X_1, c=Y_1,X_1\neq Y_1 , V_1\neq X_1), \{r_2@1,2, \ r_2@7,8, \ r_1@9,10 \}\rangle_{12} \rrarrow _{\omega'_t}
\end{array}
$
\vspace*{0.2cm}\\
$
\begin{array}{l}
\langle (\text{path}(I_3,U_3)\#13, I_3=X_3, U_3=Y_3,
\text{root} (V_2)\#9, \text{same} (V_2,X_2)\#10,  \\ \text{success}(I_2)\#12,
 \text{path}(V_1,Y_1)\#6,
\text{edge}(a,b)\#3, \text{edge}(a,c)\#4, \text{edge}(c,d)\#5), \\
(a=U_3, b=Z_3, a=I_3,b=J_3, J_3=Z_3 , I_3=U_3,
V_2=I_2, V_2=J_2, X_2=L_2,V_2=I_1, \\
X_2=J_1, Y_2=L_1, X_2\neq Y_2, V_2=X _2, I_1=V_1, J_1=V_1, L_1=X_1, a=V_1,
 b=X_1,\\c=Y_1,X_1\neq Y_1 , V_1\neq X_1), \{r_2@1,2, \ r_2@7,8, \ r_1@9,10, \ r_4@11,3, \ r_3@13 \}\rangle_{13} \rrarrow _{\omega'_t}^*
 \end{array}
$
\vspace*{0.2cm}\\
$
\begin{array}{l}
\langle (\text{path}(I_3,U_3)\#13,
\text{root} (V_2)\#9, \text{same} (V_2,X_2)\#10,  \text{success}(I_2)\#12,\\
 {\bf path (V_1,Y_1)\#6},
\text{edge}(a,b)\#3, {\bf edge (a,c)\#4}, \text{edge}(c,d)\#5), \\
(I_3=X_3, U_3=Y_3,a=U_3, b=Z_3, a=I_3,b=J_3, J_3=Z_3 , I_3=U_3,
V_2=I_2, \\ V_2=J_2, X_2=L_2,V_2=I_1,
X_2=J_1, Y_2=L_1, X_2\neq Y_2, V_2=X _2, I_1=V_1, J_1=V_1, \\ L_1=X_1, a=V_1,
b=X_1,
  c=Y_1,X_1\neq Y_1 , V_1\neq X_1), \\
  \{r_2@1,2, \ r_2@7,8, \ r_1@9,10, \ r_4@11,3, \ r_3@13 \}\rangle_{13} \rrarrow _{\omega'_t}
  \end{array}
$
\vspace*{0.2cm}\\
$
\begin{array}{l}
\langle (\text{path}(I_4,U_4)\#14, I_4=X_4, U_4=Y_4,
\text{path}(I_3,U_3)\#13,
\text{root} (V_2)\#9, \\
\text{same} (V_2,X_2)\#10,  \text{success}(I_2)\#12,
\text{edge}(a,b)\#3, \text{edge}(a,c)\#4, \text{edge}(c,d)\#5), \\
(a=U_4, c=Z_4, V_1=I_4, Y_1=J_4, J_4=Z_4 , I_4=U_4,
I_3=X_3, U_3=Y_3,a=U_3, b=Z_3, \\
a=I_3,
b=J_3, J_3=Z_3 , I_3=U_3,
V_2=I_2, V_2=J_2, X_2=L_2,V_2=I_1, X_2=J_1, Y_2=L_1,\\
X_2\neq Y_2,
 V_2=X _2, I_1=V_1, J_1=V_1, L_1=X_1, a=V_1,
b=X_1,
  c=Y_1,X_1\neq Y_1 , V_1\neq X_1), \\
  \{r_2@1,2, \ r_2@7,8, \ r_1@9,10, \ r_4@11,3, \ r_3@13 , \ r_4@4,6, \ r_3@14\}\rangle_{14} \rrarrow _{\omega'_t}^*
  \end{array}
$
\vspace*{0.2cm}\\
$
\begin{array}{l}
\langle (\text{path}(I_4,U_4)\#14,
\text{path}(I_3,U_3)\#13,
\text{root} (V_2)\#9, \text{same} (V_2,X_2)\#10,  \\
\text{success}(I_2)\#12,\text{edge}(a,b)\#3, \text{edge}(a,c)\#4, \text{edge}(c,d)\#5), \\
(I_4=X_4, U_4=Y_4,a=U_4, c=Z_4, V_1=I_4, Y_1=J_4, J_4=Z_4 , I_4=U_4, I_3=X_3,\\
 U_3=Y_3,a=U_3,
 b=Z_3, a=I_3,b=J_3, J_3=Z_3 , I_3=U_3,
V_2=I_2, V_2=J_2,\\  X_2=L_2,V_2=I_1,
X_2=J_1, Y_2=L_1,
 X_2\neq Y_2, V_2=X _2, I_1=V_1, J_1=V_1, L_1=X_1,
\\
 a=V_1,b=X_1,
  c=Y_1,X_1\neq Y_1 , V_1\neq X_1), \\
  \{r_2@1,2, \ r_2@7,8, \ r_1@9,10, \ r_4@11,3, \ r_3@13 , \ r_4@4,6, \ r_3@14\}\rangle_{14} \not \rrarrow _{\omega'_t}\\
\end{array}
$
\end{small}
\end{example}

From the previous transition systems we can obtain a notion
of observable property of CHR computations that will be used in order
to prove the correctness of our unfolding rule.
The notion of "observable property" usually identifies the relevant property that one is interested in
observing as the result of a computation. In our case, we use the notion of qualified answer,
originally introduced in \cite{Fru98}: Intuitively this is the constraint obtained as the result of a non-failed computation,
including both built-in constraints and CHR constraints which have not been "solved" (i.e. transformed
by rule applications into built-in constraints).
Formally qualified answer are defined as follows.

\begin{definition}[{\sc Qualified answers}]  Let $P$ be a CHR program and let
$G$ be a goal. The set $\mathcal{QA}_P(G)$ of
qualified answers for the query $G$ in the program $P$ is defined
as follows:
$$
\hspace*{-0.2cm}\begin{array}{ll}
\mathcal{QA}_P(G) =  \{\exists_{-Fv(G)}(chr (K)\wedge D) \mid & \mathcal{CT} \not\models D \leftrightarrow {\tt false} \mbox{ and } \\
& \langle
G,\emptyset,\texttt{true}, \emptyset \rangle_1
\rightarrow^*_{\omega_t}
\langle \emptyset,   K, D, T\rangle_n\not\rightarrow_{\omega_t}\} .
\end{array}
$$
\end{definition}

Analogously, we can define the qualified answer of an annotated
program.

\begin{definition}[{\sc Qualified answers for annotated programs}]
Let $P$ be an annotated CHR program and let $G$ be a goal. The set $\mathcal{QA'}_P(G)$ of qualified answers
for the query $G$ in the annotated program $P$ is defined as follows:
$$
\begin{array}{ll}
\mathcal{QA'}_P(G) = \{\exists_{-Fv(G)}(chr (K)\wedge D) \mid & \mathcal{CT} \not\models D \leftrightarrow {\tt false} \mbox{ and } \\
& \langle I(G),
\texttt{true}, \emptyset\rangle_m\rightarrow^*_{\omega'_t}
\langle   K, D, T\rangle_n\not\rightarrow_{\omega'_t}\}.
\end{array}
$$
\end{definition}

The previous two notions of qualified answers are equivalent,
as shown by the proof (in the Appendix) of the following proposition.
This fact will be used to prove the correctness of the unfolding.

\begin{proposition}\label{prop:nequality}
Let $P$ and $Ann(P)$ be respectively a CHR program and its annotated version.
Then, for every goal $G$,
$$\mathcal{QA}_{P}(G) = \mathcal{QA'}_{Ann(P)}(G)$$
holds.
\end{proposition}

\section{The unfolding rule}\label{sec:unfolding}

In this section, we define the \emph{unfold operation} for CHR
simpagation rules. As a particular case, we obtain also
unfolding for  simplification and propagation rules, as these can
be seen as particular cases of the former.

The unfolding allows to replace  a  conjunction $S$ of constraints
(which can be seen as a procedure call) in the body of a rule $cl_r$
by the body of a rule  $cl_v$,  provided that the head of $cl_v$ matches
with $S$ (when considering also the instantiations provided by the built-in constraints in
the guard and in the body of the rule $cl_r$). More precisely, assume that
the built-in constraints in the guard and in the body
of the rule $cl_r$ imply that the head $H$ of $cl_v$,
instantiated by a substitution $\theta$, matches with the
conjunction $S$ in the body of  $cl_r$. Then, the unfolded rule is
obtained from $cl_r$ by performing the following steps:  1) the new
guard in the unfolded rule is the conjunction of the guard of $cl_r$
with the guard of $cl_v$, the latter instantiated by $\theta$ and
without  those constraints that are entailed by the built-in
constraints which are in
$cl_r$; 2) the body of $cl_v$ and the equality
$H = S$ are added to the body of $cl_r$; 3) the conjunction of constraints $S$ can
be removed, partially removed or left in the body of  the unfolded
rule, depending on the fact that $cl_v$ is a simplification, a
simpagation or a propagation rule, respectively; 4) as for the
local token store  $T_r$ associated with every rule $cl_r$, this is
updated consistently during the unfolding operations in order to
avoid that a propagation rule is used twice to unfold the same
sequence of constraints.

Before giving the formal definition of the unfolding rule, we illustrate the above steps by means of the following example.

\begin{example}\label{ex:banca}
Consider the following program $P$, similar to that one given in \cite{SS08}, which describes the rules for updating a bank account and for performing the money transfer. We write the program by using  the standard syntax, namely without using the local token store and the identifiers in the body of rules, since there are no annotated propagation rules. The program $P$ consists of the following three rules
\begin{small}$$
\begin{array}{l}
r_1@b(Acc1, Bal1), b(Acc2, Bal2), t(Acc1,Acc2, Amount) \Leftrightarrow Acc1\neq Acc2 \mid \\
\hspace*{3cm} b(Acc1, Bal1), b(Acc2, Bal2), w(Acc1, Amount), d(Acc2, Am)\\
r_2@b(Acc, Bal), d(Acc, Am)\Leftrightarrow b(Acc,B), B= Bal+Am\\
r_3@b(Acc', Bal'), w(Acc', Am')\Leftrightarrow  Bal' > Amount' \mid b(Acc', B'), B'=Bal'-Am'\\
\end{array}
$$
\end{small}
where the three rules identified by $r_1, r_2$, and $r_3$ are called $cl_{r_1}, cl_{r_2}$, and $cl_{r_3}$, respectively. The predicate names are abbreviations: $b$ for balance, $d$ for deposit, $w$ for withdraw and $t$ for
transfer.

Now, we unfold the rule $cl_{r_1}$ by using the  rule $cl_{r_2}$ and we  obtain the new clause $cl'_{r_1}$:
\begin{small}$$
\begin{array}{l}
r_1@b(Acc1, Bal1), b(Acc2, Bal2), t(Acc1,Acc2, Amount) \Leftrightarrow Acc1\neq Acc2 \mid\\
\hspace*{3.5cm}b(Acc1, Bal1),w(Acc1, Amount), b(Acc, B),\\
\hspace*{3.5cm}  B= Bal+Am, Acc2=Acc, Bal2=Bal, Amount=Am.

\end{array}
$$
\end{small}
Next, we unfold the rule $cl'_{r_1}$ by using the  rule $cl_{r_3}$ and we can obtain the new clause $cl''_{r_3}$
\begin{small}$$
\begin{array}{l}
r_3@b(Acc1, Bal1), b(Acc2, Bal2), t(Acc1,Acc2, Amount) \Leftrightarrow \\
\hspace*{0.8cm} Acc1\neq Acc2 , Bal1 > Amount \mid \hspace*{-0.2cm} \begin{array}[t]{l}
 b(Acc, B), B= Bal+Am, Acc2=Acc, Bal2=Bal,\\
Amount=Am, b(Acc', B'), B'=Bal'-Am',\\
Acc1=Acc', Bal1=Bal', Amount=Am'.
\end{array}
\end{array}$$
\end{small}
\end{example}

Before formally defining the unfolding, we need to define a
function which removes the useless tokens from the token store.

\begin{definition}\label{def:clean}
Let $  B$ be an identified goal and let $T$ be a token set,
\[clean: Goals \times Token \rightarrow Token,\] is defined as
follows: $clean (  B,T)$ deletes from $T$ all the tokens for
which at least one identifier is not present in the identified
goal $  B$. More formally
$$\begin{array}{l}
        clean(  B, T)= \{t\in T\mid  t=r@i_1, \ldots, i_k \mbox{ and }
        i_j\in id(  B), \mbox{ for each } j \in [1,k]\}.
\end{array}
$$
\end{definition}

\begin{definition}[{\sc Unfold}]\label{def:unf}
Let $P$ be an annotated CHR program  and let $cl_r, cl_v\in P$ be the two
following annotated rules
$$
\begin{array}{rcl}
r@H_1\backslash H_2 &\Leftrightarrow&  D\,|\,  K,   S_1,   S_2, C; T \mbox{ and}\\
v@H_1'\backslash H_2' &\Leftrightarrow & D '\,|\,   B;
T'
\end{array}
$$
respectively, where $C$ is the conjunction of all the built-in constraints in the body
of $cl_r$. Let  $\theta$ be a substitution such that $dom(\theta) \subseteq Fv(H_1', H_2')$ and
$\mathcal{CT}  \models (C \wedge D) \rightarrow chr(S_1,S_2)= (H_1',
H_2')\theta$. Furthermore let
$m$ be the greatest identifier which appears in the
rule $cl_r$ and let $(B_1, T_1, m_1)=inst(B, T',m)$.
Then, the \emph{unfolded} rule is:
$$r@ H_1\backslash H_2
\Leftrightarrow D, (D''\theta)\, |\,   K,  S_1,
B_1,C, chr(S_1,S_2)= (H_1', H_2'); T''$$ where $v @id (
S_1,S_2) \not \in T$, $V=\{d\in D'\mid \mathcal{CT} \models C\wedge D\rightarrow d\theta\}$,
$D''= D'\backslash V$, $Fv(D''\theta) \cap Fv(H_1', H_2')\theta \subseteq Fv(H_1,H_2)$, the
constraint $(D, (D''\theta))$ is satisfiable and
\begin{itemize}
    \item if $H_2'=\epsilon$ then $T''= T \cup T_1 \cup\{v @id (
S_1)\}$
    \item if $H_2'\neq \epsilon$ then $T''=clean((  K,  S_1) , T) \cup T_1$.
    \end{itemize}
\end{definition}

Note that  $V\subseteq
D'$ is the greatest set of built-in constraints such that $\mathcal{CT}  \models C\wedge D\rightarrow d\theta$ for each $d\in V$. Moreover,  as shown in the following, all the results in the paper are independent from the choice of the substitution $\theta$ which satisfies the conditions of Definition~\ref{def:unf}. Finally, we use  the function $inst$  (Definition~\ref{definst})
in order to increment the value of the
identifiers associated with atoms in the unfolded rule. This allows
us to distinguish the new identifiers introduced in the unfolded
rule from the old ones. Note also that the condition on the token
store is needed to obtain a correct rule.  Consider for example a ground annotated
program
\[\begin{array}{lll}
    P= & \{ & r_1@ h \Leftrightarrow   k\#1  \\
     &  & r_2@k \Rightarrow   s\#1 \\
     &  & r_3@s,s\Leftrightarrow   q\#1 \, \}
  \end{array}
\]
where
the three rules identified by $r_1, r_2$, and $r_3$ are called $cl_{r_1}, cl_{r_2}$, and $cl_{r_3}$, respectively\footnote{Here and in the following examples, we use an identifier and also a name for a rule. The reason for this is that after having performed an unfolding we could have different rules labeled by the same identifier. Moreover, we omit the token stores if they are empty.}. Let
$h$ be  the start goal. In this case, the unfolding could change
the semantics if the token store was not used. In fact, according
to the semantics proposed in Table~\ref{omega-t} or
~\ref{tab:operational-semantics}, we have that the goal $h$ has only the qualified answer
$(k,s)$.
\comment{following
computation: $  h\rightarrow^{cl_{r_1}}
k\rightarrow^{cl_{r_2}}  k,
s\not\rightarrow_{\omega_t}$.} On the other hand,  considering an
unfolding without the update of  the token store, one would have
$r_1@ h\Leftrightarrow   k \#1\stackrel{\mbox{\tiny{unfold using
$cl_{r_2}$}}}{\longrightarrow} r_1@ h\Leftrightarrow   k\#1,   s\#2
\stackrel{\mbox{\tiny{unfold using
$cl_{r_2}$}}}{\longrightarrow}\sout{r_1@h\Leftrightarrow   k\#1,
  s\#2,   s\#3}\stackrel{\mbox{\tiny{unfold using
$cl_{r_3}$}}}{\longrightarrow}r_1@h\Leftrightarrow   k\#1,   q\#4$.
So, starting from the constraint $h$ we could obtain the qualified answer $(k,q)$, that
is not possible in the original program (the rule obtained after
the wrongly applied unfolding rule is underlined).

As previously mentioned, the unfolding rules for simplification
and propagation can be obtained as particular cases of
Definition~\ref{def:unf}, by setting $H_1'=\epsilon$ and  $H_2'
=\epsilon$, respectively, and by considering accordingly the
resulting unfolded rule.
\comment{In the following examples we will use
$\odot$ to denote both  $\Leftrightarrow$ and $\Rightarrow$.}

\begin{example}\label{ex:gen_adam}
Consider the program $P$ consisting of the following four rules
$$
\begin{array}{l}
r_1@f(X, Y), f(Y, Z), f(Z, W)\Leftrightarrow g(X, Z)\#1,f(Z, W)\#2,gs(Z, X)\#3\\
r_2@g(U, V), f(V, T)\Leftrightarrow gg(U, T)\#1\\
r_3@g(U, V), f(V, T)\Rightarrow  gg(U, T)\#1\\
r_4@g(J, L)\backslash f(L, N) \Leftrightarrow gg(J, N)\#1
\end{array}
$$
that we call $cl_{r_1}$, $cl_{r_2}$, $cl_{ r_3}$, and $cl_{ r_4}$, respectively. This program
deduces information  about genealogy. Predicate $f$ is considered as
father, $g$ as grandfather, $gs$ as grandson and $gg$ as
great-grandfather. The following rules are such that we can unfold
some constraints in the body of $cl_{r_1}$ using the rule $cl_{r_2}$,  $
cl_{r_3}$, and $cl_{ r_4}$.

Now, we  unfold the body of rule $cl_{r_1}$ by using the simplification rule
$cl_{r_2}$. We use the $inst$ function
$inst(gg(U,T)\#1,\emptyset,3) = (gg(U,T)\#4,\emptyset,4)$.
So the new unfolded rule is:
$$
\begin{array}{ll}
r_1@f(X, Y),f(Y, Z), f(Z, W)\Leftrightarrow
gs(Z,X)\#3, gg(U, T)\#4,
X=U, Z=V, W=T.
\end{array}
$$

Now, we unfold the body of $cl_{r_1}$ by using the propagation rule $cl_{r_3}$.
As in the previous case, we have that $inst(gg(U,T)\#1,\emptyset,3) = (gg(U,T)\#4,\emptyset,4)$ and then
the new unfolded rule is:
$$
\begin{array}{ll}
r_1@f(X, Y),f(Y, Z), f(Z, W)\Leftrightarrow & \hspace*{-0.3cm}g(X, Z)\#1,f(Z, W)\#2,gs(Z, X)\#3, \\
& \hspace*{-0.3cm} gg(U, T)\#4, X=U, Z=V, W=T; \{r_3@1, 2\}.
\end{array}
$$

Finally, we unfold the body of rule $cl_{r_1}$ by using the simpagation
rule $cl_{r_4}$. As before, the function
\[inst(gg(J,N)\#1,\emptyset,3) = (gg(J,N)\#4,\emptyset,4)\] is computed.
The new unfolded rule is:
$$
\begin{array}{ll}
r_1@f(X, Y),f(Y, Z), f(Z, W)\Leftrightarrow  &\hspace*{-0.3cm} g(X, Z)\#1, gs(Z,X)\#3,\\
&\hspace*{-0.3cm}  gg(J, N)\#4,X=J, Z=L, W=J.
\end{array}
$$

\end{example}

The following example considers more specialized rules  with guards which are not $\texttt{true}$.

\begin{example}\label{ex:gen_adam_refined}
Consider the program consisting of the following rules
$$
\begin{array}{l}
r_1@f(X, Y), f(Y, Z), f(Z,W)\Leftrightarrow X= Adam, Y=Seth\,|\,\\
\hspace{3.5cm}g(X, Z)\#1,f(Z, W)\#2, gs(Z, X)\#3, Z=Enosh\\
r_2@g(U, V), f(V, T) \Rightarrow  U=Adam, V=Enosh\,|\, gg(U, T)\#1, T=Kenan\\
r_3@g(J, L) \backslash f(L, N) \Leftrightarrow  J=Adam, L=Enosh\,|\, gg(J, N)\#1,N=Kenan
\end{array}
$$
that, as usual, we call $cl_{r_1}, cl_{r_2}$, and $cl_{r_3}$, respectively, and which specialize the rules introduced in Example~\ref{ex:gen_adam}
to the genealogy of Adam. That is,  here we remember that Adam was father of Seth; Seth was father
of Enosh; Enosh was father of Kenan. As before, we consider the predicate $f$ as father,
$g$ as grandfather, $gs$ as grandson and $gg$ as great-grandfather.

\comment{If we unfold $cl_{r_1}$ by using $cl_{r_2}$,
where we assume $\odot=\Leftrightarrow$, we obtain:
$$
\begin{array}{l}
r_1@f(X, Y),f(Y, Z)f(Z, W)\Leftrightarrow X=Adam, Y=Seth\,|\,gg(U, T)\#4, \\
\hspace{1cm}T= Kenan,\, gs(Z, X)\#3, \, Z=Enosh, \, X=U, Z=V, W=T.
\end{array}
$$}
If we unfold $cl_{r_1}$ by using $cl_{r_3}$ we have
$$
\begin{array}{ll}
r_1@f(X, Y),f(Y, Z)f(Z, W)\Leftrightarrow & \hspace*{-0.3cm}X=Adam, Y=Seth\,|\,  \\
\hspace{3cm}& \hspace*{-0.3cm} g(X, Z)\#1, \, gs(Z, X)\#3,\, Z=Enosh,  \\
\hspace{3cm}&  \hspace*{-0.3cm} gg(J, N)\#4 \,,N= Kenan, \, X=J, Z=L, W=N.
\end{array}
$$
Moreover, when  $cl_{r_2}$ is considered to unfold $cl_{r_1}$, we obtain
$$
\begin{array}{l}
r_1@f(X, Y),f(Y, Z), f(Z, W)\Leftrightarrow X=Adam, Y=Seth\,|\\
\hspace{3.2cm}g(X, Z)\#1,\, f(Z, W)\#2,\, gs(Z, X)\#3,  \, Z=Enosh,  \\
\hspace{3.2cm} gg(U, T)\#4, \, T= Kenan,\,X=U, Z=V, W=T; \{r_2@1, 2\}.
\end{array}
$$

Note that  $U=Adam, \, V=Enosh$, which is the guard of the rule
$cl_{r_2}$, is not
added to the guard of the unfolded rule  because $U=Adam$ is
entailed by the guard of $cl_{r_1}$ and $V=Enosh$ is entailed
by the built-in constraints in the body of $cl_{r_1}$, by considering also the binding provided by the parameter passing (analogously for $cl_{r_3}$).
\end{example}

\begin{example}\label{ex:chiara2}
The program $P'$ of the Example~\ref{chiara1} is obtained from the program $Ann(P)$ of Example~\ref{chiara} by adding to $Ann(P)$
 the clauses resulting from the unfolding of the clause $r_2$ with $r_1, \, r_2$ and $r_3$ and from the unfolding of the clause $r_4$ with $r_3$. It is worth noticing  that the use of the unfolded clauses allows to decrease the number of Apply tansition steps in the successful derivation.
\end{example}

The following result states the correctness of our unfolding rule. The proof is in the Appendix.

\begin{proposition}\label{prop:equality}
Let $P$ be an annotated CHR program with
$cl_r, cl_v\in P$. Let $cl'_r$  be the result
of the unfolding of $cl_r$ with respect to $cl_v$ and let $P'$ be the program
obtained from $P$ by adding rule $cl'_r$. Then, for every goal $G$,
$\mathcal{QA'}_{P'}(G) = \mathcal{QA'}_P(G)$ holds.
\end{proposition}

Since the previous result is independent from the choice of the particular substitution
$\theta$ which satisfies the conditions of Definition~\ref{def:unf}, we can choose any such a substitution in order to define the unfolding.

Using the semantic equivalence of a CHR program and its annotated
version, we obtain  also the following corollary which shows the equivalence between a CHR program
and its annotated and unfolded version.

\begin{corollary}\label{prop:QAunf}
Let $P$ and $Ann(P)$ be respectively a CHR program and its annotated version.
Moreover let $cl_r, cl_v\in Ann(P)$ be CHR annotated rules such that
$cl'_r$  is the result
of the unfolding of $cl_r$ with respect to $cl_v$ and $P'= Ann(P)\cup \{cl'_r\}$.
Then, for every goal $G$,
$\mathcal{QA}_{P}(G) = \mathcal{QA'}_{P'}(G)$.
\end{corollary}
\begin{proof}
The proof follows from
Proposition~\ref{prop:nequality}
and Proposition~\ref{prop:equality}.
\end{proof}

\section{Safe rule replacement}\label{sec:safty-rule-deletion}

The previous result shows that we can safely add to
a program $P$ a rule resulting from the unfolding, while preserving the
semantics of $P$ in terms of qualified answers.
However, when a rule $cl_r \in P$ has been unfolded producing
the new rule $cl'_r$, in some cases we would also like to replace $cl_r$ by $cl'_r$ in $P$,
since this could improve the efficiency of the resulting program.
Performing such a replacement while preserving the semantics
is in general a very difficult task.

In the case of CHR this is mainly
due to three problems. The first one is the presence of guards in the rules.
Intuitively, when
unfolding a rule $r$ by using a rule $v$ (i.e. when replacing in
the body of $r$ a ``call'' of a procedure by its definition $v$)
it could happen that some guard in $v$ is not satisfied
``statically'' (i.e. when performing the unfold), even though it
could become satisfied at run-time when the  rule $v$ is actually
used. If we move the guard of $v$ in the unfolded version of $r$
we can then loose some computations, because the guard 
in $v$ is moved before the atoms in the body of $r$ (and those atoms could 
instantiate and satisfy the guard). In other words, the overall guard in the unfolded rule has been strengthened, which 
means that the rule applies in fewer cases. This implies that if we want to preserve the meaning
of a program in general we cannot replace the rule $r$ by its unfolded
version. Suitable conditions can be defined in order to allows such a replacement, as we do later.
The second source of difficulties consists in the pattern matching
mechanism which is used by the CHR computation. According to this mechanism, when rewriting a goal $G$ by a rule $r$ only the variables in the head of $r$ can be instantiated (to become equal to the terms in $G$).  Hence, it could happen that statically the body of a rule $r$ is not instantiated enough to perform the pattern matching involved in the unfolding,  while it could become instantiated at run-time in the computations. Also in this case replacing $r$ by its unfolded version in general is not correct. 
Note that this is not a special case of the first issue, indeed if we cannot (statically) perform the pattern matching we do not unfold the rule  
while if we move the pattern matching to the guard we could still unfold the rule (under suitable conditions).

Finally, we have the problem of the multiple heads.
In fact, let $B$ be the body of a rule $r$ and let $H$ be the (multiple)
head of a rule $v$, which can be used to unfold $r$: we cannot be sure
that at run-time all the atoms in $H$ will be used to rewrite $B$,
since in general $B$ could be in a conjunction with other atoms even
though the guards are satisfied. 	\comment{This technical point, that one can legitimately find
obscure now, will be further clarified in Section~\ref{sec:safty-rule-deletion}.}
Note that the last point does not mean that the answers of the transformed program are a subset of those of the original one, since by deleting some computations we could introduce in the transformed program new qualified answers which were not in the original program. 
This is a peculiarity of CHR and it is different from what happens in Prolog.

The next
subsection clarifies these three points by using some examples.

\subsection{Replacement problems}

As previously mentioned, the first problem in replacing a rule by its unfolded version concerns the anticipation of
the guard of the rule $cl_v$ (used to unfold the rule $cl_r$) in the guard of $cl_r$
(as we do in the unfold operation). In fact, as shown by the following example, this could lead to the loss of some
computations,  when the unfolded rule $cl'_r$  is used rather than the
original rule $cl_r$.

\begin{example}\label{esempio:mau}
Let us consider the program
\[\begin{array}{rllll}
  P= \{ &r_1  @ p(Y)  \Leftrightarrow q(Y), s(Y)\\
  & r_2 @ q(Z) \Leftrightarrow  Z=a \,|\, \texttt{true}\\
& r_3 @ s(V) \Leftrightarrow  V=a
&\}\\\end{array}
\]
where we do not consider the identifiers (and the local token store) in
the body of rules, because we do not have propagation rules in
$P$.

The unfolding  of $r_1  @ p(Y)  \Leftrightarrow q(Y), s(Y)$  by using the rule $r_2 @ q(Z) \Leftrightarrow  Z=a \,|\, \texttt{true} $ returns the new
rule $r_1  @ p(Y)  \Leftrightarrow  Y=a  \,|\, s(Y), Y=Z $. Now the
program
\[\begin{array}{rllll}
  P'= \{
  &r_1  @ p(Y)  \Leftrightarrow  Y=a  \,|\, s(Y), Y=Z \\
  & r_2 @ q(Z) \Leftrightarrow  Z=a \,|\, \texttt{true}\\
& r_3 @ s(V) \Leftrightarrow  V=a
 &\}\\
\end{array}
\]
is not semantically equivalent to $P$ in terms of qualified
answers. In fact, given the goal $G= p(X)$ we have $(X=a )\in
\mathcal{QA'}_P(G)$, while $(X=a) \not \in \mathcal{QA'}_{P'}(G).$
\end{example}

The second problem is related to the pattern matching used in CHR computations. In fact,
following Definition~\ref{def:unf}, there are some matchings that could become
possible only at run-time, and not at compile time, because a stronger (as a first order formula)
built-in constraint store is needed. Also in this case, a rule elimination could lead to lose
possible answers as illustrated in the following example.

\begin{example}\label{ex:matching}
Let us consider the program
\[
\begin{array}{lcll}P&=\{&
r_1@p(X, Y) \Leftrightarrow q(Y, X)&\\
&&r_2@q(W, a) \Leftrightarrow W=b&\\
&&r_3@q(J, T) \Leftrightarrow J=d&\}
\end{array}
\]
where, as before, we do not consider the identifiers and the token store in the
body of rules, because we do not have propagation rules in $P$.
Let $P'$ be the program where the unfolded rule
$r_1@p(X, Y) \Leftrightarrow Y=J, X=T, J=d $, obtained by using $r_3@q(J, T) \Leftrightarrow J=d$ in $P$,
substitutes the original one (note that other unfoldings are not possible, in
particular the rule $r_2@q(W, a) \Leftrightarrow W=b$ can  not be used to unfold $r_1@p(X, Y) \Leftrightarrow q(Y, X)$)
\[
\begin{array}{lclrl}P'&=\{&
r_1@p(X, Y) \Leftrightarrow Y=J, X=T, J=d &\\
&&r_2@q(W, a) \Leftrightarrow W=b&\\
&&r_3@q(J, T) \Leftrightarrow J=d&\}.
\end{array}
\]
Let $G=p(a, R)$ be a goal. We can see that
$(R=b)\in\mathcal{QA'}_P(G)$ and  $(R=b)\not\in\mathcal{QA'}_{P'}(G)$ because,
with the considered goal (and consequently
the considered built-in constraint store) $r_2@q(W,a) \Leftrightarrow W=b$ can fire in $P$ but can not fire in
$P'$.

The third problem is related to multiple heads. In fact, the
unfolding that we have defined  assumes that the head of a rule
matches completely with the body of another one, while in general,
during a CHR computation, a rule can match with constraints
produced by more than one rule and/or  introduced by the initial
goal. The following example illustrates this point.

\begin{example}\label{ex:unicatesta}
Let us consider the program
\[\begin{array}{rlll}
  P= \{ &r  @ p(Y)  \Leftrightarrow  q(Y), h(b)\\
  & v @ q(Z), h(V) \Leftrightarrow  Z=V & \}\\
\end{array}
\]
where we do not consider the identifiers and the token store in
the body of rules, as usual.

The unfolding  of $r  @ p(Y)  \Leftrightarrow  q(Y), h(b)$ by using  $v @ q(Z), h(V) \Leftrightarrow  Z=V$ returns the new
rule
$$r @ p(Y) \Leftrightarrow  Y=Z, V=b, Z=V.$$ Now the program
\[\begin{array}{rlll}
  P'= \{ &r @ p(Y)  \Leftrightarrow Y=Z, V=b, Z=V \\
  & v  @  q(Z), h(V) \Leftrightarrow   Z=V & \}\\
\end{array}
\]
where we substitute the original rule by its  unfolded version is
not semantically equivalent to $P$. In fact, given the goal $G=
p(X), h(a), q(b)$, we have that  $(X=a)\in \mathcal{QA'}_P(G)$, while $(X=a) \not \in
\mathcal{QA'}_{P'}(G).$
\end{example}

\end{example}

\subsection{A condition for safe rule replacement}

We have identified some conditions which ensure that we can safely replace the
original rule $cl_r$ by its unfolded version while maintaining the
qualified answers semantics. Intuitively, this holds when: 1)  the
constraints of the body of  $cl_r$ can
 be rewritten only by CHR rules such that all the atoms in the head contain the same set of variables;
2) there exists no rule $cl_v$ which can be fired by using
a part of constraints introduced in the body of $cl_r$ plus some
other constraints; 3) all the rules that can be applied at run-time to the body of the original rule $cl_r$, can also be applied at
transformation time.
Before defining formally these conditions, we need some further
notations. First of all, given a rule $cl_r$ we define
two sets.

The first one contains a set of pairs: for each pair the first
component is a rule that can be used to unfold $cl_r$, while the second one is the
sequence of the identifiers of the atoms in the body of $cl_r$ which
are used in the unfolding.

The second set contains all the rules that can be used for the
{\em partial unfolding} of $cl_r$; in other words, it is the set of rules that can fire by
using at least an atom in the body of $cl_r$ and necessarily
some other CHR and built-in constraints. Moreover, such a set contains
also the rules that can fire, when an opportune built-in constraint
store is provided by the computation, but that cannot be unfolded.

\begin{definition}\label{def:Pposeneg}

Let $P$ be an annotated CHR program and let $cl_r, cl_v$ be the following two annotated rules
$$
\begin{array}{rcl}
r@H_1\backslash H_2 &\Leftrightarrow&  D\,|\,  A; T \mbox{ and}\\
v@H_1'\backslash H_2' &\Leftrightarrow & D '\,|\,   B; T'
\end{array}
$$
such that $cl_r, cl_v\in P$ and  $cl_v$ is renamed
apart with respect to $cl_r$. We define $U^{+}$ and $U^{\#}$ as follows:
\begin{enumerate}
    \item\label{uno} $(cl_v, (i_1, \ldots, i_n)) \in$ $ U^{+}_P(cl_r )$
 if and only if
$cl_r $ can be
unfolded with $cl_v$ (by Definition~\ref{def:unf}) by using the sequence
of the identified atoms in $  A$ with identifiers $(i_1,
\ldots, i_n)$.
    \item\label{due} $cl_v\in  U^{\#}_P(cl_r )$ if and only if at least one of the
following conditions holds: \\
\begin{enumerate}
\item\label{ai} there exists $(A_1,A_2) \subseteq A$ and
     a built-in constraint $C$ such that
    $Fv(C) \cap Fv(cl_v) = \emptyset$, the constraint
    $D \wedge C$ is satisfiable,
    $\mathcal{CT}  \models (D\wedge C) \rightarrow \exists _{cl_v}
((chr(A_1,A_2)=(H'_1,  H'_2))\wedge D')$, $v @id
(A_1,A_2) \not \in T$, and $(cl_v, id
(A_1,A_2)) \not  \in$ $ U^{+}_P(cl_r)$ \\
\item\label{bi} or there exist $k \in A$,
    $h\in H_1'\uplus H_2'$ and
     a built-in constraint $C$ such that
    $Fv(C) \cap Fv(cl_v) = \emptyset$, the constraint
    $D \wedge C$ is satisfiable,
    $\mathcal{CT}  \models (D\wedge C) \rightarrow \exists _{cl_v}
((chr(k)=h)\wedge D')$, and there exists no $(A_1,A_2) \subseteq A$ such that
$v @id (A_1,A_2) \not \in T$ and
$\mathcal{CT}  \models (D\wedge C \wedge (chr(k)=h)) \rightarrow (chr(A_1,A_2)=(H'_1,  H'_2))$.
\\
\end{enumerate}
\end{enumerate}
\end{definition}
Some explanations are in order here.

The set $U^{+}$ contains all the couples composed by those rules that can be used to unfold a
fixed rule $cl_r$, and the identifiers of the constraints considered in the unfolding,
introduced in Definition~\ref{def:unf}.

Let us consider now the set $U^\#$.
The conjunction of built-in constraints $C$
represents a generic set of built-in constraints
(such a set naturally can be equal to every possible built-in constraint store that can
be generated by a real computation before the application of rule $cl_v$); the
condition $Fv(C) \cap Fv(cl_v) = \emptyset$
is required to avoid free variable capture,
it represents the renaming (with fresh variables) of a rule $cl_v$ with respect to the computation
before the use of the $cl_v$ itself in an \textbf{Apply'} transition;
the condition $v@id(  A_1,   A_2)\not \in T$ avoids trivial non-termination due the propagation rules;
the conditions $\mathcal{CT}  \models (D\wedge C) \rightarrow \exists _{cl_v}
((chr(  A_1,   A_2)=(H'_1,  H'_2))\wedge D')$ and
$\mathcal{CT}  \models (D\wedge C) \rightarrow \exists _{cl_v}
((chr(k)=h)\wedge D')$ secure that a strong enough built-in constraint
is provided by the computation, before the application of rule $cl_v$;
finally, the condition
$(cl_v, id (  A_1,   A_2))\not\in
U_P^{+}(cl_r)$ is required
to avoid to consider the rules that can be correctly unfolded in the body of $cl_r$.
There are two kinds of rules that are added to $U^\#$.
The first one, due to Condition~\ref{ai} in Definition~\ref{def:Pposeneg},
indicates a matching substitution problem similar to that one described  in Example~\ref{ex:matching}. The second kind, due to Condition~\ref{bi} in Definition~\ref{def:Pposeneg}, indicates a multiple
heads problem similar to that one in Example~\ref{ex:unicatesta}.
Hence, as we will see in Definition~\ref{def:nsafedel}, in
order to be able to correctly perform the unfolding, the set $U^\#$ must be empty.

Note also that if $U^{+}_P(cl_r)$ contains a pair, whose first component is
a rule with a multiple head and such that the atoms in the head contain different sets of variables, then by definition,
$U^{\#}_P(cl_r)\neq \emptyset$ (Condition~\ref{bi} of Definition~\ref{def:Pposeneg}).

The following definition introduces a notation for the set obtained by unfolding a rule with (the rules in) a program.

\begin{definition}
Let $P$ be an annotated CHR program and assume that
$cl \in P$,
\[Unf_P(cl)\] is the set of all annotated rules obtained by unfolding
the rule $cl$ with a rule in $P$, by using Definition~\ref{def:unf}.
\end{definition}

We can now give the central definition of this section.

\begin{definition}[{\sc Safe rule replacement}]\label{def:nsafedel}
Let $P$ be an annotated CHR program and let $cl_r \in P$ be the annotated rule $r@H_1\backslash H_2
\Leftrightarrow  D\,|\,  A; T$, such that the following
holds
\begin{enumerate}
\item[i)] $U^{\#}_P(cl_r ) =\emptyset$,
\item[ii)]  $U^{+}_P(cl_r) \neq \emptyset$ and
\item[iii)] for each
$r@ H_1\backslash H_2 \Leftrightarrow D'\, |\,   A'; T' \in
   Unf_P (cl_r)$
we have that
$\mathcal{CT}  \models D \leftrightarrow D'$.
\end{enumerate}
Then, we say that the rule $cl_r$
can be safely replaced (by its unfolded version) in $P$.
\end{definition}

Condition $\bf i) $ of the previous definition implies that
$cl_r $ can be
safely replaced in $P$ only if:
\begin{itemize}
\item $U^{+}_P(cl_r )$ contains only pairs, whose 
first component is a rule such that each atom in the head contains the same set of variables;

\item  a sequence of identified atoms of
body of the rule $cl_r$ can be used to fire a rule $cl_v$
only if $cl_r$ can be unfolded with $cl_v$  by using the same sequence
of the identified atoms.
\end{itemize}

Condition {\bf ii)} states that there exists at least one rule for unfolding
the rule $cl_r$.

Condition {\bf iii)}  states that each annotated rule obtained by the
unfolding of $cl_r$ in $P$ must have a guard equivalent to that one of $cl_r$: in fact
the condition  $\mathcal{CT}  \models D \leftrightarrow D'$ in {\bf iii)} avoids the
problems discussed in Example~\ref{esempio:mau}, thus allows the moving (i.e. strengthening)
of the guard in the unfolded rule.

Note that Definition~\ref{def:nsafedel} is independent from the particular substitution $\theta$ chosen in Definition~\ref{def:unf} in order to define the unfolding of the rule
\[\begin{array}{rcl}
r@H_1\backslash H_2 &\Leftrightarrow&  D\,|\,  K,   S_1,   S_2, C; T \mbox{ with respect to}\\
v@H_1'\backslash H_2' &\Leftrightarrow & D '\,|\,   B; T'
\end{array}
\]
In fact, let us assume that there exist two substitution $\theta$ and $\gamma$ which satisfy the conditions of Definition~\ref{def:unf}.
Then $\mathcal{CT}  \models (C\wedge D) \rightarrow (d\theta \leftrightarrow d\gamma)$ for each $d \in D'$.
Therefore, if $V=\{d\in D'\mid \mathcal{CT} \models C\wedge D\rightarrow d\theta\}$ and
$W=\{d\in D'\mid \mathcal{CT} \models C\wedge D\rightarrow d\gamma\}$, we have that $V=W$ and then
$D''=D\setminus V= D \setminus W$.
Now, it is easy to check that Condition {\bf iii)} follows if and only if $D'' \theta=D''\gamma=\emptyset$.

The following is an example of a safe replacement.
\begin{example}\label{ex:safe1}
Consider the program $P$ consisting of the following four rules
$$
\begin{small}
\begin{array}{l}
r_1@p (X,Y, Z)  \Leftrightarrow r(b,b,Z)\#1, s(Z,b,a)\#2, q(X, f(Z), a)\#3, r(g(X,b), f(a), f(Z))\#4;\emptyset\\
r_2@q(V,U,W), r(g(V,b), f(W), U)\Leftrightarrow W=a \mid s(V,U,W)\#1, r(U,U,V)\#2;\emptyset\\
r_3@r(M,M,N), s (N,M,a)\Leftrightarrow  p(M,N,N)\#1;\emptyset\\
r_4@s(L,J,I)\Rightarrow  I=L;\emptyset\\
\end{array}
\end{small}
$$
where the four rules identified by $r_1, r_2, r_3$, and $r_4$ are called $cl_{r_1}, cl_{r_2}, cl_{r_3}$, and $cl_{r_4}$, respectively.
By Definition \ref{def:Pposeneg}, we have that
\[
\begin{array}{l}
U^{+}_P(cl_{r_1}) =\{cl_{r_2}@3,4, \, cl_{r_3}@1,2, \, cl_{r_4}@2\}\\
 U^{\#}_P (cl_{r_1}) =\emptyset.
\end{array}
\]

Moreover
$$
\begin{small}
\begin{array}{l}
Unf_P (cl_{r_1})= \\
\begin{array}{llll}
\{& r_1@p (X,Y, Z)  \Leftrightarrow & \hspace*{-0.2cm}r(b,b,Z)\#1, s(Z,b,a)\#2,s(V,U,W)\#5, r(U,U,V)\#6,
\\
&& \hspace*{-0.2cm}X=V, U=f(Z), W=a;\emptyset\\
&r_1@p (X,Y, Z)  \Leftrightarrow & \hspace*{-0.2cm}q(X, f(Z), a)\#3, r(g(X,b), f(a), f(Z))\#4, p(M,N,N)\#5,  \\
&& \hspace*{-0.2cm}M=b, N=Z ;\emptyset\\
&r_1@p (X,Y, Z)  \Leftrightarrow & \hspace*{-0.2cm}r(b,b,Z)\#1, s(Z,b,a)\#2, q(X, f(Z), a)\#3,
 \\
&&\hspace*{-0.2cm}  r(g(X,b), f(a), f(Z))\#4, I=L, Z=L, b=J, a=I;
\{cl_{r_4}@2\}&\}
\end{array}
\end{array}
\end{small}
$$

Then  $cl_1$ can be safely replaced in $P$ according to Definition~\ref{def:nsafedel} and then we obtain
\[\begin{array}{l}
  P_1  =  (P\setminus   \{cl_1\} ) \, \cup
   Unf_{P}(cl_1),
\end{array}
\]
where $P_1$ is the program
$$
\begin{small}
\begin{array}{lll}
\{& r_1@p (X,Y, Z)  \Leftrightarrow  \hspace*{-0.2cm}
\begin{array}[t]{l}
r(b,b,Z)\#1, s(Z,b,a)\#2,s(V,U,W)\#5, r(U,U,V)\#6,
\\
X=V, U=f(Z), W=a;\emptyset
\end{array}\\
& r_1@p (X,Y, Z)  \Leftrightarrow  \hspace*{-0.2cm}\begin{array}[t]{l} q(X, f(Z), a)\#3, r(g(X,b), f(a), f(Z))\#4, p(M,N,N)\#5, \\
 M=b, N=Z ;\emptyset
 \end{array}
 \\
& r_1@p (X,Y, Z)  \Leftrightarrow \hspace*{-0.2cm}\begin{array}[t]{l}
r(b,b,Z)\#1, s(Z,b,a)\#2, q(X, f(Z), a)\#3, \\
 r(g(X,b), f(a), f(Z))\#4, I=L, Z=L, b=J, a=I;
\{cl_{r_4}@2\}
\end{array}\\
& r_2@q(V,U,W), r(g(V,b), f(W), U)\Leftrightarrow W=a \mid s(V,U,W)\#1, r(U,U,V)\#2;\emptyset\\
& r_3@r(M,M,N), s (N,M,a)\Leftrightarrow  p(M,N,N)\#1;\emptyset\\
& r_4@s(L,J,I)\Rightarrow  I=L;\emptyset & \}
\end{array}
\end{small}
$$
\end{example}

We can now provide the result which shows the correctness of our safe replacement rule. The proof is in the Appendix.

\begin{theorem}\label{theo:n1completeness}
Let $P$ be an annotated program,  $cl$ be a rule in $P$ such that
$cl$ can be safely replaced
in $P$ according to Definition~\ref{def:nsafedel}. Assume also that
\[\begin{array}{l}
  P'  =  (P\setminus   \{cl\} ) \, \cup
   Unf_{P}(cl).
\end{array}
\]

Then $\mathcal{QA'}_{P}(G)=\mathcal{QA'}_{P'}(G)$ for any arbitrary
goal $G$.
\end{theorem}

Of course, the previous result can be applied to a sequence of program transformations.
Let us define such a sequence as follows.

\begin{definition}[{\sc U-sequence}]\label{def:uno}
Let $P$ be an annotated CHR program. An \emph{U-sequence} of programs
starting from $P$ is a sequence of annotated CHR programs $P_0, \ldots,
P_n$, such that
\[
\begin{array}{lll}
  P_0 & = & P  \mbox{ and }\\
  P_{i+1}& = &  (P_i \setminus   \{cl^i \} ) \, \cup  Unf_{P_i}(cl^i), \\
\end{array}
 \]
 where $i \in [0,n-1]$, $cl^i \in P_i$ and can be safely replaced in $P_i$.
\end{definition}

\begin{example}\label{ex:safe2}
Let us to consider the program $P_1$ of Example~\ref{ex:safe1}. The clause $cl_2$ can be safely replaced in $P_1$ according to Definition~\ref{def:nsafedel} and then we obtain
\[\begin{array}{l}
  P_2  =  (P_1\setminus   \{cl_2\} ) \, \cup
   Unf_{P_1}(cl_2),
\end{array}
\]
where $P_2$ is the program

$$
\begin{small}
\begin{array}{lll}
\{&\hspace*{-0.2cm} r_1@p (X,Y, Z)  \Leftrightarrow  \hspace*{-0.2cm}
\begin{array}[t]{l}
r(b,b,Z)\#1, s(Z,b,a)\#2,s(V,U,W)\#5, r(U,U,V)\#6,
\\
X=V, U=f(Z), W=a;\emptyset
\end{array}\\
&\hspace*{-0.2cm} r_1@p (X,Y, Z)  \Leftrightarrow  \hspace*{-0.2cm}\begin{array}[t]{l} q(X, f(Z), a)\#3, r(g(X,b), f(a), f(Z))\#4, p(M,N,N)\#5, \\
 M=b, N=Z ;\emptyset
 \end{array}
 \\
& \hspace*{-0.2cm}r_1@p (X,Y, Z)  \Leftrightarrow \hspace*{-0.2cm}\begin{array}[t]{l}
r(b,b,Z)\#1, s(Z,b,a)\#2, q(X, f(Z), a)\#3, \\
 r(g(X,b), f(a), f(Z))\#4, I=L, Z=L, b=J, a=I;
\{cl_{r_4}@2\}
\end{array}\\
&\hspace*{-0.2cm} r_2@q(V,U,W), r(g(V,b), f(W), U)\Leftrightarrow W=a \mid p(M,N,N)\#3, V=N, U=M ;\emptyset\\
&\hspace*{-0.2cm} r_2@q(V,U,W), r(g(V,b), f(W), U)\Leftrightarrow W=a \mid \hspace*{-0.2cm}
\begin{array}[t]{l}s(V,U,W)\#1, r(U,U,V)\#2,  V=L, \\
 U=J, W=I , I=L;\{cl_{r_4}@1\}
\end{array}\\
&\hspace*{-0.2cm} r_3@r(M,M,N), s (N,M,a)\Leftrightarrow  p(M,N,N)\#1;\emptyset\\
&\hspace*{-0.2cm} r_4@s(L,J,I)\Rightarrow  I=L;\emptyset & \hspace*{-0.3cm} \}.
\end{array}
\end{small}
$$
\end{example}

Then, from  Theorem~\ref{theo:n1completeness} and
Proposition~\ref{prop:nequality}, we have the
following.

\begin{corollary}\label{lemma:ncompleteness}
Let $P$ be a program and let $P_0, \ldots, P_n$ be an U-sequence starting from $Ann(P)$. Then
$\mathcal{QA}_{P}(G)=\mathcal{QA'}_{P_n}(G)$ for any arbitrary goal
$G$.
\end{corollary}

\subsection{Confluence and Termination}\label{sec:confluence&termination}

In this section, we prove that our unfolding preserves termination
provided that one considers normal derivations. These
are the derivations in which the  \textbf{Solve} (\textbf{Solve'})
transitions are applied as soon as possible, as specified by Definition~\ref{def:ND}.
Moreover, we prove that our unfolding preserves also confluence, provided that one considers only non-recursive unfoldings.

We first need to introduce the concept of built-in
free configuration: This is a configuration which has no built-in constraints in the first component. \comment{or has an unsatisfiable built-in store.}

\begin{definition}[{\sc Built-in free configuration}]\label{def:BFS}
Let  $\sigma=\la G, S,D, T\ra_o\in {\it Conf_t}$ ($\sigma=\la  G , D, T\ra_o\in {\it Conf'_t}$).
The configuration $\sigma$ is built-in free  if
$G$ is a multiset of (identified) CHR-constraints.
\end{definition}

Now, we can introduce the concept of normal derivation.

\begin{definition}[{\sc Normal derivation}]\label{def:ND}
Let  $P$ be a (possibly annotated) CHR program  and let
$\delta$ be a derivation in $P$.
We say that $\delta$ is normal if, for each configuration  $\sigma$ in $\delta$, a transition
\textbf{Apply} (\textbf{Apply'})
is used on $\sigma$ only if $\sigma$ is built-in free.
\end{definition}

Note that, by definition, given a CHR program $P$, $\mathcal{QA}(P)$ can be calculated by considering only normal derivations and analogously for an annotated CHR program $P'$.

\begin{definition}[{\sc Normal Termination}]
A CHR program $P$ is called \emph{terminating}, if there are no infinite derivations.
A (possibly annotated) CHR program $P$ is called  \emph{normally terminating}, if there are no infinite normal derivations.
\end{definition}

The following result shows that normal termination is preserved by unfolding with the safe replacement condition. The proof is in the Appendix.

\begin{proposition}[{\sc Normal Termination}]\label{prop:termination} Let $P$ be a CHR program and
let $P_0, \ldots, P_n$ be an U-sequence starting from $Ann(P)$. $P$ satisfies
normal termination if and only if $P_n$  satisfies normal termination.
\end{proposition}

When standard termination is considered  rather than  normal termination,
the previous result does not hold, due to  the guard elimination in the unfolding. This is shown by the following example.

\begin{example}\label{ex:termination-problem}
Let us consider the following program:
\[
\begin{array}{lcll}
P&=\{&r_1@p(X)\Leftrightarrow X=a, q(X)&\\
&&r_2@q(Y)\Leftrightarrow Y=a \mid r(Y)&\\
&&r_3@r(Z)\Leftrightarrow Z=d \mid p(Z)&\}
\end{array}
\]
where we do not consider the identifiers and the token store in the
body of rules (because we do not have propagation rules in $P$).
Then, by using $$r_2@q(Y)\Leftrightarrow Y=a \mid r(Y)$$ to unfold $r_1@p(X)\Leftrightarrow  X=a, q(X)$ (with replacement) we obtain the following program $P'$:
\[
\begin{array}{lcll}
P'&=\{&r_1@p(X)\Leftrightarrow  X=a, X=Y, r(Y)&\\
&&r_2@q(Y)\Leftrightarrow Y=a \mid r(Y)&\\
&&r_3@r(Z)\Leftrightarrow Z=d \mid p(Z)&\}.
\end{array}
\]
It is easy to check that the program $P$ satisfies the (standard) termination.
On the other hand, considering the program $P'$ and the start goal $(V=d,p(V))$, the following state can be reached
$$\la  (X=a, p(Z)\#3),(V=d, V=X, X= Y, Y=Z), \emptyset \ra_3$$
where rules $r_1@p(X)\Leftrightarrow X=a, X=Y, r(Y)$ and $r_3@r(Z)\Leftrightarrow Z=d \mid p(Z)$ can be applied infinitely many times if the built-in constraint $X=a$ is not moved by the \textbf{Solve'} rule into the built-in store. Hence, we have non-termination.\end{example}

The  next property we consider is confluence. This property guarantees that any computation for a goal
results in the same final state, no matter which of the applicable
rules are applied \cite{AF04,Fru04}.

We first give the following definition which introduces some specific notation for renamings of indexes.

\begin{definition}\label{def:renaming}
Let $j_1,
\ldots, j_o$ be distinct identification values.
\begin{itemize}
  \item A renaming of identifiers is a
substitution of the form $[j_1/i_1, \ldots, j_o/i_o]$, where  $i_1, \ldots,
i_o$ is a permutation of $j_1, \ldots, j_o$.
  \item Given an expression $E$ and a renaming of identifiers $\rho= [j_1/i_1, \ldots, j_o/i_o]$, $E \rho$
  is defined as the expression obtained from $E$ by
substituting each occurrence of the identification value $j_l$
with the corresponding $i_l$, for $l \in [1,o]$
  \item If $\rho$ and $\rho'$ are renamings of identifiers, then $\rho \rho'$ denotes the renaming of identifiers such that for each expression $E$, $E (\rho \rho')= (E \rho)\rho'$.
\end{itemize}
 We will use $\rho, \
\rho', \ldots $ to denote renamings.
\end{definition}

Now, we need the
following definition introducing a form of equivalence between configurations, which is a slight modification of that one in \cite{RBF09},
since considers a different form of configuration  and, in particular, also the presence of the token store.
Two configurations are equivalent if they have the same logical reading and the same rules are applicable to these configurations with the same results. By an abuse of notation, when it is clear from the context, we will write $\equiv_V$ to denote two equivalence relations in
${\it Conf_t}$ and in ${\it Conf'_t}$ with the same meaning.

\begin{definition}\label{def:naltraeq}
Let $V$ be a set of variables
The equivalence $\equiv _V$ between configurations in ${\it Conf_t}$  is the smallest equivalence relation  that satisfies the following conditions.

\begin{itemize}
\item $\langle d \wedge G,   S,C, T\rangle_{n}\equiv _V \langle G,   S,d \wedge C, T\rangle_{n}$,

\item $\langle G,   S, X=t \wedge C, T\rangle_{n}\equiv _V \langle G [X/t], S [X/t],X=t \wedge C, T\rangle_{n}$,

\item Let $X,Y$ be variables such that $X,Y \not \in V$ and $Y$ does not occur in $G,   S$ or $c$.
$\langle G,   S, C, T\rangle_{n}\equiv _V \langle G [X/Y], S [X/Y], C[X/Y], T\rangle_{n}$,

\item If $W=Fv(C) \setminus (Fv(G,S) \cup V)$, $U=Fv(C') \setminus (Fv(G,S) \cup V)$, and
$CT \models \exists _W C  \leftrightarrow \exists _{U} C'$ then
$\langle G,   S,C, T\rangle_{n}\equiv _V \langle G,   S, C' , T\rangle_{n}$,

\item $\langle  G,
  S, {\tt false} , T\rangle_n \equiv _V \langle  G',
  S', {\tt false} , T'\rangle_m$,

\item $\langle G,   S,C, T\rangle_{n}\equiv _V \langle G,   S\rho ,C, T\rho \rangle_{m}$
for each renaming of identifiers $\rho$ such that for each $i \in  \mathit id (S\rho) \cup  \mathit  id (T\rho)$ we have that $i<m$,

\item
$\langle G,   S,C, T\rangle_{n}\equiv _V \langle G,   S,C,  clean(S, T)\rangle_{n}$.
 \end{itemize}

\end{definition}

We can define the equivalence $\equiv _V$ between configurations in ${\it Conf'_t}$ in an analogous way.
\begin{definition}\label{def:n1altraeq}
Let $V$ be a set of variables
The equivalence $\equiv _V$ between configurations in ${\it Conf'_t}$ is the smallest equivalence relation  that satisfies the following conditions.

\begin{itemize}
\item $\langle d \wedge G,  C, T\rangle_{n}\equiv _V \langle G, d \wedge C, T\rangle_{n}$,

\item $\langle G,  X=t \wedge C, T\rangle_{n}\equiv _V \langle G[X/t], X=t \wedge C, T\rangle_{n}$,

\item Let $X,Y$ be variables such that $X,Y \not \in V$ and $Y$ does not occur in $G$ or $c$.
$\langle G,   C, T\rangle_{n}\equiv _V \langle G [X/Y], C[X/Y], T\rangle_{n}$,

\item If $W=Fv(C) \setminus (Fv(G) \cup V)$, $U'=Fv(C') \setminus (Fv(G) \cup V)$, and
$CT \models \exists _W C  \leftrightarrow \exists _{u} C'$ then
$\langle G, C, T\rangle_{n}\equiv _V \langle G,  C' , T\rangle_{n}$,

\item $\langle  G, {\tt false} , T\rangle_n \equiv _V \langle  G',
  {\tt false} , T'\rangle_m$,

\item $\langle G, C, T\rangle_{n}\equiv _V \langle G\rho ,C, T\rho \rangle_{m}$
for each renaming of identifiers $\rho$ such that for each $i \in  \mathit id (G\rho) \cup  \mathit id (T\rho)$
we have that $i\leq m$,

\item
$\langle G,C, T\rangle_{n}\equiv _V \langle G,C,  clean( G, T)\rangle_{n}$.
 \end{itemize}
\end{definition}

By definition of $\equiv_V$, it is straightforward to check that
if $\sigma, \sigma' \in {\it Conf_t} ({\it Conf'_t})$, $V$ is a set of variables, and $\sigma \equiv _V \sigma'$ then the following holds
\begin{itemize}
\item if $W \subseteq V$ then $\sigma \equiv _W \sigma'$ and
\item if $X \not \in Fv(\sigma)\cup Fv(\sigma')$ then $\sigma \equiv _{V\cup \{X\}} \sigma'$.
\end{itemize}
\comment{\begin{definition}\label{def:vecchiaaltraeq}
Let
$\sigma=$\langle G,   S,c, T\rangle_{n}$, \sigma_2 \in {\it Conf_t} ({\it Conf'_t})$ and let $V$ be a set of variables. $\sigma_1\equiv_{V}\,\sigma_2$ if at least one of the following holds:
\begin{itemize}
  \item $\sigma_1$ and $\sigma_2$ are both failed configurations;
  \item $\sigma_1$ and $\sigma_2$ are identical up to renaming of variables not in $V$ and of identifiers, up to logical equivalence of built-in constraints and up to cleaning of the token store (that is, up to deleting from the token store all the tokens for
which at least one identifier is not present in the set of
identified CHR constraints, see Definition~\ref{def:clean}). \end{itemize}
\end{definition}
}
We now introduce the concept of confluence which is a slight modification of that one in \cite{RBF09}, since it considers also the cleaning of the token store.

In the following $\mapsto^{*}$ means either $\rrarrow_{\omega_t}^{*}$ or $\rrarrow_{\omega'_t}^{*}$.

\begin{definition}[{\sc Confluence}]\label{def:Conf}
A CHR [annotated] program is \emph{confluent} if for any state $\sigma$ the following holds:
if $\sigma\mapsto^{*} \sigma_1$ and $\sigma\mapsto^{*} \sigma_2$ then there exist states $\sigma_f'$
and $\sigma_f''$ such that $\sigma_1 \mapsto^{*}\sigma_f'$ and $\sigma_2\mapsto^{*}\sigma_f''$, where $\sigma_f'\equiv_{Fv(\sigma)}\,\sigma_f''$.
\end{definition}

Now, we prove that our unfolding preserves confluence, provided that one considers only non-recursive unfolding. These
are the unfoldings such that a clause $cl$ cannot be used in order to unfold  $cl$  itself.

When safe rule replacement is considered  rather than  non-recursive safe rule replacement (see Definition~\ref{def:nrsafedel}),
the confluence is not preserved. This is shown by the following example.

\begin{example}\label{ex:cofluence-problem}
Let us consider the following program:
\[
\begin{array}{lcll}
P&=\{&r_1@p\Leftrightarrow q&\\
&&r_2@p\Leftrightarrow r&\\
&&r_3@r\Leftrightarrow r,s&\\
&&r_4@q\Leftrightarrow r,s&\}
\end{array}
\]
where we do not consider the identifiers and the token store in the
body of rules (because we do not have propagation rules in $P$).
Then, by using $r_3$ to unfold $r_3$ itself (with safe rule replacement) we obtain the following program $P'$:
\[
\begin{array}{lcll}
P'&=\{&r_1@p\Leftrightarrow q&\\
&&r_2@p\Leftrightarrow r&\\
&&r_3@r\Leftrightarrow r,s,s&\\
&&r_4@q\Leftrightarrow r,s&\}.
\end{array}
\]
It is easy to check that the program $P$ is confluent.
On the other hand, considering the program $P'$ and the start goal $p$, the following  two states can be reached
$$\sigma=\la  (r\#3, s\#4, s\#5),\texttt{true}, \emptyset \ra_5 \mbox{ and } \sigma'=\la  (r\#3, s\#4),\texttt{true}, \emptyset \ra_4$$
and
there exist no states $\sigma_1$
and $\sigma_1'$ such that $\sigma \rrarrow_{\omega'_t}^{*}\sigma_1$ and $\sigma'\rrarrow_{\omega'_t}^{*}\sigma_1'$ in $P'$, where $\sigma_1\equiv_{\emptyset}\,\sigma_1'$.\end{example}

Note that the program in previous example is not terminating. 
We cannot consider a terminating program here, since for such a program (weak) safe rule replacement would allow to preserve confluence.
Now, we give the definition of non-recursive safe rule replacement.

\begin{definition}[{\sc Non-recursive safe rule replacement}]\label{def:nrsafedel}
Let $P$ be an annotated CHR program and let $cl_r \in P$ be an annotated rule such that
$cl_r$
can be safely replaced (by its unfolded version) in $P$.
We say that $cl_r$
can be non-recursively safely replaced (by its unfolded version) in $P$ if
for each $(cl_v, (i_1, \ldots, i_n)) \in$ $ U^{+}_P(cl_r )$, we have that
$cl_v \neq cl_r$.
\end{definition}

The following is the analogous of Definition~\ref{def:uno}, where non-recursive safe rule replacement is considered.

\begin{definition}[{\sc NRU-sequence}]\label{def:nruno}
Let $P$ be an annotated CHR program. An \emph{NRU-sequence} of programs
starting from $P$ is a sequence of annotated CHR programs $P_0, \ldots,
P_n$, such that
\[
\begin{array}{lll}
  P_0 & = & P  \mbox{ and }\\
  P_{i+1}& = &  (P_i \setminus   \{cl^i \} ) \, \cup  Unf_{P_i}(cl^i), \\
\end{array}
 \]
 where $i \in [0,n-1]$, $cl^i \in P_i$ and can be non-recursively safely replaced in $P_i$.
\end{definition}

\begin{theorem} Let $P$ be a CHR program and let $P_0, \ldots, P_n$
be an NRU-sequence starting from $P_0= Ann(P)$. $P$ satisfies confluence if and only if $P_n$
satisfies confluence too.
\end{theorem}

\section{Weak safe rule replacement}\label{sec:normally -rep}
In this subsection, we consider only programs which are normally terminating and confluent.
For this class of programs we give a condition for rule replacement which is much weaker than that one used in the previous section and which still allows one to preserve the qualified answers semantics. Intuitively this new condition requires that there exists a rule obtained by the unfolding of $cl_r$ in $P$ whose guard is equivalent to that one in $cl_r$.

\begin{definition}[{\sc Weak safe rule replacement}]\label{def:wsafedel}
Let $P$ be an annotated CHR program and let $r@H_1\backslash H_2
\Leftrightarrow  D\,|\,  A; T \in P$ be a rule such that there exists
$$r@ H_1\backslash H_2 \Leftrightarrow D'\, |\,   A'; T' \in
   Unf_P (r@H_1\backslash H_2 \Leftrightarrow D\,|\,  A; T)$$
with $\mathcal{CT}  \models D \leftrightarrow D'$.

Then, we say that the rule $r@H_1\backslash H_2 \Leftrightarrow  D\,|\,  A; T $
can be weakly safely replaced (by its unfolded version) in $P$.
\end{definition}

%
%
%
%
%

\begin{example}\label{ex:ncofluence}
Let us consider the following program $P$:
\[
\begin{array}{lcll}
P_1&=\{&r_1@p(X)\Leftrightarrow q(X),s(X)&\\
&&r_2@t(a)\Leftrightarrow r(b)&\\
&&r_3@q(Y)\Leftrightarrow t(Y)&\\
&&r_4@s(a)\setminus q(a)\Leftrightarrow r(b)&\}
\end{array}
\]
where we do not consider the identifiers and the token store in the
body of rules (because we do not have propagation rules in $P$).
By Definition \ref{def:wsafedel}, $r_1$ can be weakly safely replaced (by its unfolded version) in $P$ and then we can obtain the program
\[P_1  =  (P\setminus   \{r_1\} ) \, \cup \,
   Unf_{P}(r_1),\]
   where $P_1$ is the program
\[
\begin{array}{lcll}
P_1&=\{&r_1@p(X)\Leftrightarrow s(X), t(Y),X=Y &\\
&&r_2@t(a)\Leftrightarrow r(b)&\\
&&r_3@q(Y)\Leftrightarrow t(Y)&\\
&&r_4@s(a)\setminus q(a)\Leftrightarrow r(b)&\}.
\end{array}
\]
Finally, observe that $r_1$ cannot be safely replaced (by its unfolded version) in $P$.
\end{example}

The following proposition shows that normal termination and confluence
are preserved by weak safe rule replacement. The proof is in the Appendix.

\begin{proposition}\label{prop:wterm}
Let $P$ be an annotated CHR program and let $cl \in P$ such that $cl $
can be weakly safely replaced in $P$. Moreover let
\[P'  =  (P\setminus   \{cl\} ) \, \cup \,
   Unf_{P}(cl).\]
If $P$ is normally terminating then $P'$ is normally terminating. If $P$ is normally terminating and confluent then $P'$ is confluent too.
\end{proposition}

The converse of the previous theorem does not hold, as shown by the following example.
\begin{example}\label{ex:wconfluence-problem}
Let us consider the following program:
\[
\begin{array}{lcll}
P&=\{&r_1@p(X)\Leftrightarrow q(X)&\\
&&r_2@q(a)\Leftrightarrow p(a)&\\
&&r_3@q(Y)\Leftrightarrow r(Y)&\}
\end{array}
\]
where we do not consider the identifiers and the token store in the
body of rules (because we do not have propagation rules in $P$).
Then, by using $r_3$ to unfold $r_1$ itself (with weak safe rule replacement) we obtain the following program $P'$:
\[
\begin{array}{lcll}
P'&=\{&r_1@p(X)\Leftrightarrow X=Y, r(Y)&\\
&&r_2@q(a)\Leftrightarrow p(a)&\\
&&r_3@q(Y)\Leftrightarrow r(Y).&\}
\end{array}
\]
It is easy to check that the program $P'$ satisfies the (normal) termination.
On the other hand, considering the program $P$ and the start goal $p(a)$, the following state can be reached
$$\la  (p(a)\#3),(X=a), \emptyset \ra_3$$
where rules $r_1@p(X)\Leftrightarrow q(X)$ and $r_2@q(a)\Leftrightarrow p(a)$ in $P$ can be applied infinitely many times.
Hence, we have non-(normally)termination.\end{example}

Next, we show that weak safe rule replacement transformation preserves qualified answers.

\begin{theorem}\label{prop:wqualified}
Let $P$ be a normally terminating and confluent annotated program and let $cl$ be a rule in $P$ such that
$cl $ can be weakly safely replaced
in $P$ according to Definition~\ref{def:wsafedel}. Assume also that
\[P'  =  (P\setminus   \{cl\} ) \, \cup \,
   Unf_{P}(cl).
\]
Then $\mathcal{QA'}_{P}(G)=\mathcal{QA'}_{P'}(G)$ for any arbitrary
goal $G$.
\end{theorem}

\begin{proof}
Analogously to Theorem~\ref{theo:n1completeness}, by using Proposition~\ref{prop:nequality}  we can prove that
$\mathcal{QA'}_{P}(G)=\mathcal{QA'}_{P''}(G)$ where
 \[
  P'' = P \, \cup \, Unf_{P}(cl),
\]
for any arbitrary goal $G$.

Then, to prove the thesis, we have only to prove that
\[\mathcal{QA'}_{P'}(G) = \mathcal{QA'}_{P''}(G).\]
We prove the two inclusions separately.
\begin{description}
  \item[{\bf ($\mathcal{QA'}_{P'}(G) \subseteq \mathcal{QA'}_{P''}(G)$) }]
  The proof is the same of the case $\mathcal{QA'}_{P'}(G) \subseteq \mathcal{QA'}_{P''}(G)$ of Theorem~\ref{theo:n1completeness} and hence it is omitted.

  \item[{\bf ($\mathcal{QA'}_{P''}(G) \subseteq \mathcal{QA'}_{P'}(G)$) }]
  The proof is by contradiction. Assume that there exists $Q \in \mathcal{QA'}_{P''}(G) \setminus \mathcal{QA'}_{P'}(G)$. Since, from the proof of Proposition~\ref{prop:wterm}, we can conclude that $P''$ is normally terminating and confluent, we have that $\mathcal{QA'}_{P''}(G)$ is a singleton. Moreover, since by the previous point
  $\mathcal{QA'}_{P'}(G) \subseteq \mathcal{QA'}_{P''}(G)$, we have that $\mathcal{QA'}_{P'}(G)= \emptyset$.
  This means that each normal derivation in $P'$ either is not terminating or terminates with a failed configuration. Then, by using Proposition~\ref{prop:wterm}, we have that each normal derivation in $P'$ terminates with a failed configuration. Since $P' \subseteq P''$, we have that there exist normal derivations in $P''$ which terminate with a failed configuration. Then, by Lemma~\ref{conflnormterm} and since $Q \in \mathcal{QA'}_{P''}(G)$, we have a contradiction and then the thesis holds.
\end{description}
\end{proof}

Let $cl$ be the rule $r@H_1\backslash H_2 \Leftrightarrow D\,|\,  A; T$.
Note that Proposition~\ref{prop:wterm} and Theorem~\ref{prop:wqualified} hold also if \[P'=
 (P\setminus   \{cl\} ) \, \cup \, S,\]
 where $S \subseteq Unf_{P}(cl)$  and there exists
$cl'=r@H_1\backslash H_2 \Leftrightarrow D'\,|\,  A'; T'
\in S$  such that  $\mathcal{CT}  \models D \leftrightarrow D'$.

\comment{Analogously if $cl^i$ is the rule $r@H_1\backslash H_2 \Leftrightarrow D\,|\,  A; T$ then the following Corollary~\ref{lemma:wcompleteness} holds also if
\[P_{i+1}=  (P_i \setminus   \{cl^i\}) \, \cup S_i\]
where $S_i \subseteq Unf_{P_i}(cl^i)$  and there exists
$ r@H_1\backslash H_2 \Leftrightarrow D'\,|\,  A'; T'  \in S_i$ such that $\mathcal{CT}  \models D \leftrightarrow D'$.\\}

If in Definition~\ref{def:uno} we consider weak safe rule replacement rather than safe rule replacement, then we can obtain a definition
of  WU-sequence (rather than U-sequence).
From the previous theorem and by Proposition~\ref{prop:wterm}, by using an obvious inductive argument, we can derive that the semantics (in terms of qualified
answers) is preserved in WU-sequences starting from a normally terminating and confluent annotated program,
where weak safe replacement is applied repeatedly.

%

\section{Conclusions}\label{sec:conclusion_and_future}

In this paper, we have defined an unfold operation for CHR which preserves the qualified answers of a program.

This was obtained by transforming a CHR program into an
annotated one which is then unfolded. The equivalence of the unfolded
program and the original (unannotated) one is proven by using  a slightly modified operational semantics for annotated programs.
We have then provided a condition that can be used to replace a
rule by its unfolded version, while preserving the qualified answers. We have also shown that this condition ensures that confluence and termination are preserved, provided that one considers normal derivations. Finally,
we have defined a further, weaker, condition which allows one to safely replace a rule by its unfolded version (while preserving qualified answers) for programs which are normally terminating and confluent.

There are only few other papers that consider source to source transformation
of CHR programs. \cite{Fru04}, rather than considering a generic transformation
system focuses on the specialization of rules w.r.t. a specific goal, analogously
to what happens in partial evaluation. In \cite{FH03} CHR rules are transformed
in a relational normal form, over which a source to source transformation is performed. Some form of
transformation for probabilistic CHR is considered in \cite{WFLP02}, while
guard optimization was studied in \cite{SSD05b}. Another paper which involves program transformation
for CHR is \cite{SS09}.

Both the general and the goal specific approaches are important in order
to define practical transformation systems for CHR. In fact, on the
one hand of course one needs some general unfold rule, on the other
hand, given the difficulties in removing rules from the transformed
program, some goal specific techniques can help to improve the
efficiency of the transformed program for specific classes of
goals. A method for deleting redundant CHR rules is considered
in \cite{AF04}. However, it is based on a semantic check and it
is not clear whether it can be transformed into a specific syntactic
program transformation rule.

When considering more generally the field of concurrent logic
languages, we find few papers which address the issue of program
transformation. Notable examples include  \cite{EGM01} that deals with
the transformation of  Concurrent Constraint Programming (CCP) and
\cite{UF88} that considers Guarded Horn Clauses (GHC).  The
results in these papers are not directly applicable to CHR  because
neither CCP nor GHC allow rules with multiple heads.

As mentioned in the introduction, some of the results presented here appeared in
\cite{TMG07} and in the thesis  \cite{Tac08}. However, it is worth noticing that
the conditions for safe rule replacement that we have presented in Section 5
and the content of Section 6 are original contributions of this paper.
In particular, differently from the conditions given in \cite{TMG07} and \cite{Tac08},
the conditions defined in Section 5 allow us to perform rule replacement also when rules with multiple heads
are used for unfolding a
given rule. This is a major improvement, since CHR rules have naturally multiple heads.

The results obtained in the current article can
be considered as a first step
in the direction of defining a transformation system for CHR programs,
based on unfolding. This step could be extended in several directions:
First of all, the unfolding operation could be extended to take
into account also the constraints in the propagation part of the head
of a rule.
Also, we could extend to CHR some of the other transformations,
notably folding \cite{TS84} which has already been applied to CCP in \cite{EGM01}.
Finally, we would like to investigate from a practical perspective
to what extent program transformation can improve the
performance of the CHR solver. Clearly, the application of an
unfolded rule avoids some computational steps (assuming  that
unfolding is done at the time of compilation, of course). However, the
increase in the number of program rules produced by unfolding could eliminate this improvement.

Here, it would probably be important to consider some unfolding
strategy, in order to decide which rules have to be unfolded.


An efficient unfolding strategy could also incorporate in particular
probabilistic or statistical information. The idea would be to only unfold 
CHR rules which are used often and leave those which are used only 
occasionally unchanged in order to avoid an unnecessary increase in the
number of program rules. This approach could be facilitated by probabilistic 
CHR extensions such as the ones as presented for example
in \cite{WFLP02}\ and \cite{CHRPrism}.  Extending the results of this paper
to probabilistic CHR will basically follow the lines and ideas presented here.
The necessary information
which one would need to decide whether and in which sequence to 
unfold CHR rules could obtain experimentally, e.g. by profiling, or formally 
via probabilistic program analysis. One could see this as a kind
of {\em speculative} unfolding.

\bibliographystyle{acmtrans}
\bibliography{thesis}

\newpage

\appendix

 \section{Proofs}

In this appendix, we give the proofs of some of the results contained in the paper.

\subsection{Equivalence of the two operational semantics}

Here, we provide the proof of  Proposition~\ref{prop:nequality}.
To this aim we first introduce some preliminary notions and lemmas.

Then, we define two configurations (in the two different transition systems)
equivalent when they are essentially the same up to renaming of identifiers.

\begin{definition}[{\sc Configuration equivalence}]\label{def:PLQA}
 Let
 $\sigma=\la(H_1,C),  H_2, D, T\ra_n \in {\it Conf_t}$ be a configuration in the transition system $\omega_t$ and let  $\sigma'=\la( K,C), D, T'\ra_m \in {\it Conf'_t}$ be a configuration in the transition system $\omega'_t$.

$\sigma$ and $\sigma'$ are \emph{equivalent} (and we write
$\sigma\approx\sigma'$) if:
\begin{enumerate}
\item\label{uno-e}  there exist $K_1$ and $K_2$, such that $K =  K_1 \uplus K_2$, $H_1= chr(K_1)$ and $chr(H_2) = chr(K_2)$,
\item\label{due-e} for each $l \in id (  K_1)$, $l$ does not occur in $T'$,
\item\label{tre-e} there exists a renaming of identifier $\rho$ s.t. $H_2\rho=K_2$ and
$T\rho=T'$.
\end{enumerate}
\end{definition}
Condition~\ref{uno-e} grants that $\sigma$ and $\sigma'$ have equal
CHR constraints, while Condition~\ref{due-e} ensures that no propagation rule is applied to constraints in $\sigma'$ corresponding to constrains in $\sigma$ that are not previously
introduced in the CHR store. Finally, condition~\ref{tre-e} requires that there exists a renaming of identifiers such that
the identified CHR constraints  and the tokens of $\sigma$
and the ones associated with them in $\sigma'$ are equal, after the renaming.

The following result shows the equivalence of the two introduced
semantics proving the equivalence of intermediate configurations.

\begin{lemma}\label{lemma:intermequiv}
Let $P$ and $Ann(P)$ be respectively a CHR program and its annotated version.
Moreover, let
$\sigma \in {\it Conf_t}$ and let  $\sigma' \in {\it Conf'_t}$ such that
$\sigma\approx\sigma'$.
Then, the following holds
\begin{itemize}
  \item there exists a derivation $\delta = \sigma \rrarrow^*_{\omega_t} \sigma_1$ in $P$ if and only if
there exists a derivation $\delta' =\sigma' \rrarrow^*_{\omega'_t} \sigma'_1$ in $Ann(P)$ such  $\sigma_1\approx\sigma'_1$
  \item the number of \textbf{Solve} (\textbf{Apply}) transition steps in $\delta$ and the number of \textbf{Solve'} (\textbf{Apply'}) transition steps in $\delta'$ are equal.
\end{itemize}
\end{lemma}
\begin{proof}
 We show that any transition step from any configuration in one system can be imitated from a (possibly empty) sequence of  transition steps from an equivalent configuration in the other system to achieve an equivalent configuration.
 Moreover there exists a \textbf{Solve} (\textbf{Apply}) transition step in $\delta$ if and only if there exists a  \textbf{Solve'} (\textbf{Apply'}) transition step in $\delta'$.

 Then, the proof follows by a straightforward inductive argument.

Let $\sigma=\la(H_1,C),  H_2, D, T\ra_n \in {\it Conf_t}$
and let  $\sigma'=\la(  K,C), D, T'\ra_m \in {\it Conf'_t}$ such that
$\sigma\approx\sigma'$.
By definition of $\approx$,  there exist $K_1$ and $K_2$ and a renaming $\rho$ such that
\begin{equation}
  \label{7giu3} K =K_1\uplus K_2, \, H_1 = chr(K_1),  \, chr (H_2) = chr(  K_2), \,  H_2 \rho=   K_2\mbox{ and } T\rho=T'.
\end{equation}

\begin{description}
\item[Solve and Solve':] they move a built-in constraint from the Goal store or the Store
respectively to the built-in constraint store. In this case, let $C=C' \uplus \{c\}$. By definition of the two transition systems
\[\begin{array}{l}
\sigma \rrarrow_{\omega_t}^{Solve}
\la(H_1,C'),  H_2, D \wedge c, T\ra_n \mbox{ and }
\sigma ' \rrarrow_{\omega'_t}^{Solve'}
\la(  K,C'), D \wedge c, T'\ra_m.
\end{array}
\]
By definition of $\approx$, it is easy to check that $\la(H_1,C'),  H_2, D \wedge c, T\ra_n \approx
\la(  K,C'), D \wedge c, T'\ra_m$.
\\

\item[Introduce:] this kind of transition there exists only in $\omega_t$ semantics and its
application labels a CHR constraint in the goal store and moves it in the CHR store.
In this case, let $H_1=H'_1\uplus \{h\}$ and

$$
\begin{array}{l}
\sigma\rrarrow_{\omega_t}^{Introduce}
\la (H'_1, C), H_2 \uplus \{h\#n\}, D, T\ra_{n+1}.
\end{array}
$$

Let $H_2'= H_2 \uplus \{h\#n\}$. By (\ref{7giu3}) and since $H_1=H'_1\uplus \{h\}$,
there exists an identified atom $h\#f \in K_1$.
Let $n'=\rho(n)$ (where $n'=n$ if $n$ is not in the domain of $\rho$).

Now, let $K'_1=K_1 \setminus \{h\#f\}$ and $K'_2= K_2 \uplus \{h\#f\}$. By (\ref{7giu3}), we have that $K = K'_1\uplus  K'_2$, $H'_1 = chr(K'_1)$ and $chr (H'_2) = chr(K'_2)$.

Moreover, by definition of $\approx$, for each $l \in id (K_1)$, $l$ does not occur in $T'$. Therefore, since by construction $K'_1 \subseteq  K_1$, we have that for each $l \in id (  K'_1)$, $l$ does not occur in $T'$.

Now, to prove that $\sigma' \approx
\la (H'_1, C),    H_2', D, T\ra_{n+1}$, we have only to prove that there exists a renaming $\rho'$, such that
$T\rho'=T'$ and $  H'_2 \rho'=   K'_2$.
We can consider the new renaming $\rho'=\rho  \{n'/f, f/n'\}$.
By definition $\rho'$ is a renaming of identifiers.

Let us start proving that
$  H'_2\rho'=  K'_2$.

We recall that
$ H_2\rho=  K_2$ by hypothesis.
Since by construction, $f \not \in id (K_2)= id(H_2\rho)$, we have that
$  H_2 \rho'=  H_2 \rho \{n'/f, f/n'\}=   H_2 \rho \{n'/f\}.$
Moreover, since by definition $n \not \in id (  H_2)$ and $n'=\rho(n)$, we have that
$  H_2 \rho \{n'/f\}=   H_2 \rho$.
By the previous observations, we have that
\[  H'_2 \rho'=   H_2 \rho
\, \uplus \, (\{h\#n\} \{n/f\}) =   K'_2.\]
Finally, we prove that $T\rho'=T'$. Since by definition of configurations in ${\it Conf_t}$,
$n$ does not occur in $T$ and $n'= \rho(n)$, we have that $T \rho'=(T \rho) \{f/n'\}=T'\{f/n'\}$, where the last equality follows by hypothesis. Moreover  since $f \in id (  K_1)$, we have that $f$ does  not occur in $T'$.
Therefore, $T'\{ f/n'\}=T'$ and then the thesis.
\\

\item[Apply and Apply':] Let $cl_r =r@F'\backslash F'' \Leftrightarrow D_1 \,|\, B, C_1 \in P$  and let $cl'_r =r@F'\backslash F'' \Leftrightarrow D_1 \,|\, \tilde B, C_1\in Ann(P)$
be its annotated version, where ${\tilde B}= I(B)$.
 The latter can
be applied to the considered configuration $\sigma'=\la(K,C), D, T'\ra_m$. In particular $F',F''$ match respectively with $ P_1$ and $ P_2$ such that
$ P_1 \uplus P_2\subseteq  K$. Without loss of generality, by using a suitable number of {\bf Introduce} steps, we can assume that
$r@F'\backslash F'' \Leftrightarrow D_1 \,|\, B, C_1 \in P$ can be applied to $\sigma=\la(H_1,C), H_2, D, T\ra_n$. In particular, considering the hypothesis $\sigma\approx\sigma'$, we can assume for $i=1,2$, there exists $ Q_i $ such that $ Q_1 \uplus  Q_2 \subseteq  H_2$, $ Q_i \rho =  P_i$ and $F',F''$ match respectively with $Q_1$ and $Q_2$.

Then, by (\ref{7giu3}), there exist $P_3$ and $ Q_3$ such that
$ Q_3 \rho =  P_3$, $ K_2=  P_1 \uplus  P_2\uplus P_3$ and
$ H_2=  Q_1 \uplus  Q_2 \uplus  Q_3$.

By construction, since $T \rho=T'$ and $( P_1, P_2)= ( Q_1, Q_2)\rho$ (and then
$chr ( P_1, P_2)=chr ( Q_1, Q_2)$), we have that

\begin{itemize}
\item $r@id( P_1, P_2) \not \in T'$ if and only if
$r@id( Q_1, Q_2) \not \in T$ and
\item $\mathcal{CT}  \models D \rightarrow
\exists _{cl'_r}(( (F', F'')=chr ( P_1,  P_2))\wedge D_1)$ if and only if
$\mathcal{CT}  \models D \rightarrow
\exists _{cl_r}( ((F', F'')=chr ( Q_1,  Q_2))\wedge D_1)$.
\end{itemize}

Therefore, by definition of {\bf Apply} and of {\bf Apply'}
\[\sigma \rrarrow^{Apply}_{\omega_t} \la \{H_1,C\}\uplus \{B, C_1 \},( Q_1,  Q_3),
((F',  F'')=chr( Q_1, Q_2))\wedge D_1 \wedge D, T_1 \ra_n\]
if and only if
\[\sigma' \rightarrow^{Apply'}_{\omega'_t}  \la( K_1, P_1, P_3 ,C, B', C_1),
((F',  F'')=chr( P_1,  P_2))\wedge D_1 \wedge  D, T'_1 \ra_o\]
where
\begin{itemize}
  \item $T_1=T \cup \{{\it r}@id(Q_1,Q_2)\}$,
  \item $( B',\emptyset ,o )= inst(\tilde B,\emptyset ,m)$ and
  \item $T_1'=T'\cup \{r @id(P_1,P_2)\}$.
\end{itemize}

Let $\sigma_1= \la \{H_1,C\}\uplus \{B, C_1 \},( Q_1,  Q_3),
((F',  F'')=chr( Q_1, Q_2))\wedge D_1 \wedge D, T_1 \ra_n$ and
$\sigma'_1= \la( K_1, P_1,  P_3 ,  B',C, C_1),
((F',  F'')=chr( P_1,  P_2))\wedge D_1 \wedge D, T'_1 \ra_o$. \\

Now, to prove the thesis, we have to prove that
$\sigma_1 \approx \sigma'_1$.

The following holds.

\begin{enumerate}
\item  There exist $ K'_1= ( K_1,  B') $ and $ K'_2=( P_1,  P_3)$, such that $( K_1, P_1,  P_3 ,  B') = K'_1\cup  K'_2$,
    $H_1 \uplus B = chr( K'_1)$ and $chr ( Q_1, Q_3) = chr( K'_2)$.
\item Since for each $l \in id ( K_1)$, $l$ does not occur in $T'$,
 $ P_1 \subseteq  K_2$ and by definition of \textbf{Apply'} transition, we have that for each $l \in id ( K'_1)= id( K_1,  B')$, $l$ does not occur in $T'_1$,
\item By construction and since $T\rho=T'$, we have that $T_1\rho=T'_1$.
Moreover, by construction  $( Q_1, Q_3) \rho= ( P_1, P_3)=  K'_2$.
\end{enumerate}
By definition, we have that $\sigma_1 \approx \sigma'_1$ and then the thesis.
\end{description}
\end{proof}

\noindent Then, we easily obtain the following \\

\setcounter{proposition}{0}
\begin{proposition}
Let $P$ and $Ann(P)$ be respectively a CHR program and its annotated version.
Then, for every goal $G$,
$$\mathcal{QA}_{P}(G) = \mathcal{QA'}_{Ann(P)}(G)$$
holds.
\end{proposition}
\setcounter{proposition}{4}
\begin{proof}
By definition of $\mathcal{QA}$ and of $\mathcal{QA'}$,
the initial configurations of the two transition systems are equivalent. Then, the proof follows by Lemma~\ref{lemma:intermequiv}.
\end{proof}

\subsection{Correctness of the unfolding}
We prove now the correctness of our unfolding definition.

Next proposition states that qualified answers can be obtained  by considering normal derivations only for both the semantics  considered. Its proof is straightforward and hence it is omitted.

\begin{proposition}\label{prop:solonorm}
Let $P$ be CHR program  and let $P'$ an annotated CHR program.
Then
$$
\begin{array}{ll}
\mathcal{QA}_P(G) =
\{\exists_{-Fv(G)}(chr(K)\wedge d) \mid & \mathcal{CT} \not\models d \leftrightarrow {\tt false}, \\
& \delta= \langle G,\emptyset,\texttt{true}, \emptyset \rangle_1
\rightarrow^*_{\omega_t} \langle \emptyset,K, d, T\rangle_n\not\rightarrow_{\omega_t}\\
&\mbox{and $\delta$ is a normal derivation}\}\\
\end{array}
$$
and
$$\begin{array}{ll}
\mathcal{QA'}_{P'}(G) =
\{\exists_{-Fv(G)}(chr(K)\wedge d) \mid & \mathcal{CT} \not\models d \leftrightarrow {\tt false},\\
&\delta=\langle I (G),
\texttt{true}, \emptyset\rangle_m\rightarrow^{*}_{\omega'_t}
\langle   K, d, T\rangle_n\not\rightarrow_{\omega'_t}\\
&\mbox{and $\delta$ is a normal derivation}\}.
\end{array}
$$
\end{proposition}

The next proposition essentially shows the correctness of unfolding w.r.t. a derivation step.
We first define an equivalence between configurations in ${\it Conf'_{t}}$.

\begin{definition}[{\sc Configuration Equivalence}]\label{def:SE}
Let  $\sigma=\la  G , D, T\ra_o$ and  $\sigma'=\la   G', D', T'\ra_o$ be configurations  in ${\it Conf'_t}$.
$\sigma$ and $\sigma'$ are equivalent and we write $\sigma\simeq\sigma'$ if one of the following facts hold:
\begin{itemize}
  \item $\sigma$ and $\sigma'$ are both failed configurations
  \item or $G = G'$, $\mathcal{CT} \models D \leftrightarrow D'$ and $clean(G, T)= clean(G', T')$.
\end{itemize}

\end{definition}

\begin{proposition}\label{prop:servequality}
Let $cl_r, cl_v$ be annotated CHR rules and $cl'_r$  be the result
of the unfolding of $cl_r$ with respect to $cl_v$. Let $\sigma$ be a generic built-in free configuration such that we can use the transition \textbf{Apply'} with the rule $cl'_r$ obtaining the configuration $\sigma_{r'}$ and then the built-in free configuration $\sigma_{r'}^f$. Then, we can construct a derivation which uses at most the rules $cl_r$ and $cl_v$ and obtain a built-in free configuration $\sigma^f$ such that $\sigma_{r'}^f \simeq \sigma^f$.
\end{proposition}
\begin{proof}
Assume that
$$
\begin{array}{rl}
\sigma&\rrarrow^{cl'_r}\sigma_{r'}\rrarrow^{Solve^{*}}
\sigma_{r'}^f\\
&\searrow_{\, cl_r}
\sigma_r \rrarrow^{Solve^{*}}
\sigma_r^f (\rrarrow^{cl_v}
\sigma_v\rrarrow^{Solve^{*}}
\sigma_v^f)
\end{array}
$$

The labeled arrow $\rrarrow^{Solve^{*}}$ means that only {\bf Solve} transition steps are applied.
Moreover:
\begin{itemize}
           \item if $\sigma_r^f$ has the form
$\la   G, {\tt false}, T\ra$ then the derivation between the parenthesis is not present and $\sigma^f=\sigma_r^f$.
           \item the derivation between the parenthesis is present and $\sigma^f=\sigma_v^f$, otherwise.
         \end{itemize}
Let  $\sigma=\la (  H_1,  H_2,  H_3) , C, T\ra_j$ be a built-in free configuration
and let $cl_r$ and $cl_v$ be the rules
$r@H'_1\backslash H'_2 \Leftrightarrow  D_r\,|\, K,S_1,S_2, C_r; T_r $ and
$v@S_1'\backslash S_2' \Leftrightarrow  D_v\,|\, P, C_v;T_v$ respectively,
where $C_r$ is the conjunction of all the built-in constraints in the body
of $cl_r$, $\theta$ is a substitution such that $dom(\theta) \subseteq Fv(S_1',S_2')$ and
\begin{equation}\label{13marzo1}
    \mathcal{CT}  \models (D_r \wedge C_r) \rightarrow chr(S_1,S_2)=(S_1',S_2')\theta.
\end{equation}
Furthermore assume that $m$ is the greatest identifier which appears in the
rule $cl_r$ and that $inst(P, T_v,m)=(P_1, T_1, m_1)$.
Then, the \emph{unfolded} rule $cl'_r$ is:
\[
    r@ H'_1\backslash H'_2
\Leftrightarrow D_r, (D_v'\theta)\, |\,   K,  S_1,
P_1, C_r, C_v, chr(  S_1,   S_2)= (S_1', S_2'); T_{r'}
\]
where $v @id (S_1,S_2) \not \in T_r$,
$V\subseteq D_v$, $V=\{ c \mid \mathcal{CT}  \models (D_r\wedge C_r)\rightarrow c\theta\},$
$D_v'= D_v\backslash V$, $Fv(D_v'\theta) \cap Fv((S_1',S_2')\theta)\subseteq Fv(H_1',H_2')$, the
constraint $(D_r, (D_v'\theta))$ is satisfiable and
\begin{itemize}
    \item if $S_2'=\epsilon$ then $T_{r'}=T_r \cup T_1 \cup
              \{v @id (S_1)\}$
    \item if $S_2'\neq \epsilon$ then $T_{r'}=clean(( K, S_1), T_r) \cup T_1$.
    \end{itemize}

By the previous observations, we have that
    \begin{equation}\label{10dic2}
    \mathcal{CT}  \models (D_r\wedge C_r)\rightarrow V\theta,
    \end{equation}
and therefore $\mathcal{CT}  \models V\theta \leftrightarrow \exists _{-Fv(D_r\wedge C_r)}V\theta$. Then, without loss of generality, we can assume that
    \begin{equation}\label{3giu101}
    Fv(V\theta) \subseteq Fv(cl_r).
    \end{equation}
Analogously, by (\ref{13marzo1}) and since $dom(\theta) \subseteq Fv(S_1',S_2')$, we can assume that
 \begin{equation}\label{7giu5}
    Fv(chr(S_1,S_2)=(S_1',S_2')\theta) = Fv((chr(S_1,S_2)=(S_1',S_2'))\theta)\subseteq Fv(cl_r).
 \end{equation}

Moreover, since by definition $Fv(D_v'\theta) \cap Fv((S_1',S_2')\theta)\subseteq Fv(H_1',H_2')$ and $dom(\theta) \subseteq Fv(S_1',S_2')$, we have that
    \begin{equation}\label{7giu1}
    Fv(D_v'\theta) \subseteq Fv(H_1',H_2') \cup Fv(cl_v).
    \end{equation}
Let us consider the application of the rule $cl'_{r}$ to $\sigma$. By definition of the transition \textbf{Apply'}, we have that
\begin{equation}\label{10dic1}
    \mathcal{CT} \models C\rightarrow \exists_{cl'_{r}}
    ((chr(H_1, H_2)=(H'_1, H'_2))\wedge D_r\wedge (D_v'\theta))
\end{equation}
and
$$\begin{array}{l}
    \sigma_{r'}=\la (Q,C_r, C_v,chr(S_1,S_2)= (S_1', S_2')), D, T_3\ra_{j+m_1},
  \end{array}
  $$
 \noindent
 where
 \begin{itemize}
   \item $  Q=( H_1,  H_3,  Q_1)$,
   \item $\mathcal{CT}  \models  D \leftrightarrow(chr(H_1,   H_2)= (H_1', H_2')\wedge D_r\wedge (D_v'\theta)\wedge C)$,
   \item  $inst ((  K,  S_1,  P_1), T_{r'}, j)=(  Q_1, T_{r'}', j+m_1)$ and
     $T_3=T \cup T_{r'}' \cup\{r @id (H_1,H_2)\}$.
 \end{itemize}

    Therefore, by definition
    $\sigma_{r'}^f=\la   Q, C_{r'}^f, \, T_3\ra_{j+m_1},$
    where
    $$\begin{array}{ll}
       \mathcal{CT}  \models C_{r'}^f \leftrightarrow  (C_r\wedge  C_v\wedge chr(  S_1,   S_2)= (S_1', S_2') \wedge D).
     \end{array}
    $$

Let us consider now the application of
$cl_r$ to $\sigma$ and then of $cl_v$ to the $\sigma_r^f$ obtained from the previous application.
Since by construction $Fv((chr(H_1, H_2)=(H'_1, H'_2))\wedge D_r) \cap Fv(cl'_r) \subseteq Fv(cl_{r})$ and by (\ref{10dic1}), we have that
\[\mathcal{CT} \models C\rightarrow \exists_{cl_r}((chr(H_1, H_2)=(H'_1, H'_2))\wedge D_r).
\]
Therefore, by definition of the transition \textbf{Apply'}, we have that
$$\begin{array}{l}
    \sigma_{r}=\la (  Q_2, C_r),
          chr(  H_1,   H_2)= (H_1', H_2')\wedge D_r \wedge C, T_4\ra_{j+m},
  \end{array}
  $$
\noindent where
\begin{itemize}
  \item $  Q_2=(  H_1,  H_3,  K'',  S''_1,  S''_2)$,
  \item $((K'',S''_1,S''_2), T_2, j+m) = inst ((K,S_1,S_2), T_r, j)$ and
$T_4=T \cup T_2 \cup\{r @id (H_1, H_2)\}$.
\end{itemize}

Therefore, by definition $\sigma_{r}^f=\la  Q_2, C_{r}^f, \, T_4\ra_{j+m},$
  where
\begin{equation}\label{13marzo2}
\mathcal{CT}  \models C_{r}^f \leftrightarrow  C_r\wedge
     chr(H_1, H_2)= (H_1', H_2')\wedge D_r\wedge C.
\end{equation}

Now, we have two possibilities
\begin{description}
  \item[($C_{r}^f = \tt false$).] In this case, by construction, we have that $C_{r'}^f = \tt false$. Therefore $ \sigma_{r'}^f \simeq  \sigma_{r}^f$ and then the thesis.
  \item[($C_{r}^f    \neq \tt false$).]
\comment{By (\ref{10dic1}) and (\ref{13marzo2})
  \[\begin{array}{ll}
    \mathcal{CT} \models  & C_{r}^f
     \rightarrow  C_r \wedge D_r \wedge \exists _{cl'_r} (D'_v \theta ).
  \end{array}
 \]
Then, by} By (\ref{13marzo2}) and (\ref{10dic2})
(\ref{13marzo1}), we have that
\[\mathcal{CT} \models C_{r}^f \rightarrow
chr(S_1,S_2)= (S_1', S_2')\theta \wedge
V \theta.\]

Moreover, by (\ref{13marzo2}), (\ref{10dic1}) and (\ref{7giu1})
\[\begin{array}{ll}
    \mathcal{CT} \models  C_{r}^f
     \rightarrow  & \exists_{H_1',H_2', cl_v}
     (chr(H_1,H_2)=(H_1',H_2')\wedge (D'_v \theta )) \\
     &\wedge \, chr(H_1,H_2)=(H_1',H_2')
  \end{array}
 \]
and then $
    \mathcal{CT} \models   C_{r}^f
     \rightarrow  \exists_{cl_v}
     (D'_v \theta)$.
Therefore, by (\ref{3giu101}), (\ref{7giu5}) and since the rules are renamed apart,
\[\mathcal{CT} \models C_{r}^f \rightarrow \exists_{cl_v}
     (chr(S_1,S_2)= (S_1', S_2')\theta \wedge
V \theta \wedge D'_v \theta ).\]
Then, by definition of $D_v$ and since $dom(\theta) \subseteq Fv(S_1',S_2')$, we have that
$\mathcal{CT} \models   C_{r}^f \rightarrow \exists_{cl_v} ((chr(S_1,S_2)=(S_1',S_2') \wedge D_v) \theta)$.

Therefore, since  $dom(\theta) \subseteq Fv(S_1',S_2')\subseteq Fv(cl_v)$,
\[
 \mathcal{CT} \models C_{r}^f \rightarrow \exists_{cl_v} (chr(S_1,S_2)=(S_1',S_2') \wedge D_v).\]

Then, $\sigma^f_r$ is such that we can use the transition \textbf{Apply'} with the rule $cl_v$ obtaining the new configuration
$$\begin{array}{ll}
    \sigma_{v}=&\la (Q_3 ,  C_v),\, D', \, T_5\ra_{m_1},
  \end{array}
  $$
where \begin{itemize}
        \item $Q_3=(H_1,H_3,  K'',  S''_1,   P_2)$
        \item $\mathcal{CT} \models D'  \leftrightarrow (chr( S_1, S_2)= (S_1', S_2')\wedge D_v \wedge C_r\wedge chr(H_1,H_2)= (H_1',H_2')\wedge D_r \wedge C$),
        \item $inst(P, T_v, j+m) = (  P_2, T_v', m_1)$ and $T_5=T_4 \cup T_v' \cup\{v @id (
S''_1, S''_2)\}$.
      \end{itemize}

Finally, by definition, we have that
$\sigma_{v}^f=\la   Q_3,C_{v}^f, \, T_5\ra_{m_1},$
  where
  $$\begin{array}{l}
   \mathcal{CT} \models  C_{v}^f \leftrightarrow C_v \wedge D'.
  \end{array}
  $$
  By definition of $D$ and $D'$, we have that $\mathcal{CT} \models C_{r'}^f \leftrightarrow C_{v}^f$.

  If $C_{v}^f =\tt false $ then the proof is analogous to the previous case and hence it is omitted.
  Otherwise, observe that by construction,
$  Q=(  H_1,  H_3,  Q_1)$, where
$  Q_1$ is obtained from $(  K,  S_1,  P_1)$ by adding the natural $j$ to each identifier in
$(  K,  S_1)$ and by adding the natural $j+m$ to each identifier in
$  P$.
Analogously, by construction,
$  Q_3=(  H_1,  H_3,  K'',  S''_1,P_2)$, where
$(  K'',  S''_1)$  are obtained from $(  K,  S_1)$ by adding the natural $j$ to each identifier in
$(  K,  S_1)$ and $  P_2$ is obtained from $  P$ by adding the natural  $j+m$ to each identifier in $  P$.

Therefore $Q=Q_3$ and then, to prove the thesis, we have only to prove that \[clean(  Q,T_3)=clean(  Q, T_5).\]

Let us introduce the function $inst': Token\times\mathbb{N} \rrarrow \mathbb{N}$ as
the restriction of the function $inst$ to token sets and natural numbers, namely
$inst'(T,n)= T'$, where $T'$ is obtained from $T$ by
incrementing each identifier in $T$ with $n$. So, since
$T'_{r'}= inst'(T_{r'}, j)$,
$T_{r'}= T_r\cup T_1\cup \{v@id(S_1, S_2) \}$ and
$T_1=inst'(T_v, m)$, we have that

$$
\begin{array}{lcl}
T_3 & = & T\, \cup \, T'_{r'}\, \cup \,\{r@id(  H_1, H_2)\}\\
       &=& T \, \cup \, inst'(clean((K,S_1), T_r), j)\, \cup \, inst'(T_v, j+ m)
       \,\cup \\
       && inst'(\{v@id(  S_1, S_2)\}, j)\, \cup \,
       \{r@id(H_1, H_2)\}
\end{array}
$$

Analogously, since $T_4= T \cup T_2 \cup \{r@id(H_1,H_2)\}$,
$T_2 = inst'(T_{r}, j)$ and $T'_v= inst'(T_v,j+m)$, we have that
$$
\begin{array}{lcl}
T_5 &=& T_4\, \cup \, T'_v \, \cup \, \{v@id(  S''_1,S''_2)\}\\

    &=& T \, \cup \, inst'(T_r, j) \, \cup \, \{r@id(H_1,H_2)\}\,
    \cup \, inst'(T_v,j+m) \, \cup \, \\
    &&\{v@id(S''_1, S''_2)\}
    \end{array}
$$
Now, since by construction $(S''_1, S''_2)$ is obtained from $(S_1, S_2)$ by adding the natural $j$ to each identifier, we have that $inst'(\{v@id( S_1, S_2)\}, j)= \{v@id(S''_1, S''_2)\}$.
Moreover, by definition of annotated rule $id(T_r) \subseteq id (  K,   S_1,   S_2)$
and $  Q=(  H_1,  H_3,  Q_1)$, where
$  Q_1$ is obtained from $(  K,  S_1,  P_1)$ by adding the natural $j$ to each identifier in
$(  K,  S_1)$ and by adding the natural $j+m$ to each identifier in
$  P$. Then $clean(  Q, inst'(clean((  K,   S_1), T_r), j))=
clean(  Q, inst'(T_r, j))$ and then the thesis holds.
\end{description}
\end{proof}

\noindent Hence we obtain the correctness result. \\

\setcounter{proposition}{1}
\begin{proposition}
Let $P$ be an annotated CHR program with
$cl_r, cl_v\in P$. Let $cl'_r$  be the result
of the unfolding of $cl_r$ with respect to $cl_v$ and let $P'$ be the program
obtained from $P$ by adding rule $cl'_r$. Then, for every goal $G$,
$\mathcal{QA'}_{P'}(G) = \mathcal{QA'}_P(G)$ holds.
\end{proposition}
\setcounter{proposition}{6}

\begin{proof}
We prove the two inclusions separately.
\begin{description}
  \item[{\bf ($\mathcal{QA'}_{P'}(G) \subseteq \mathcal{QA'}_P(G)$)}] The proof follows from Propositions~\ref{prop:solonorm} and~\ref{prop:servequality} and by a straightforward inductive argument.
  \item[{\bf ($\mathcal{QA'}_{P}(G) \subseteq \mathcal{QA'}_{P'}(G)$)}]
  The proof is by contradiction. Assume that there exists $Q \in \mathcal{QA'}_{P}(G) \setminus \mathcal{QA'}_{P'}(G)$. By definition there exists a derivation
  \[\delta=\langle I(G),
\texttt{true}, \emptyset\rangle_m\rightarrow^{*}_{\omega'_t}
\langle   K, d, T\rangle_n\not\rightarrow_{\omega'_t}\] in $P$, such that $Q =
\exists _{-Fv(G)}(chr(   K)\wedge d)$. Since $P \subseteq P'$, we have that there exists the derivation
$\langle I(G),
\texttt{true}, \emptyset\rangle_m\rightarrow^{*}_{\omega'_t}
\langle   K, d, T\rangle_n$ in $P'$. Moreover, since $P' =P \cup \{cl'_r\}$ and by hypothesis $Q\not  \in \mathcal{QA'}_{P'}(G)$, we have that there exists a derivation step $\langle   K, d, T\rangle_n\rightarrow_{\omega'_t}
\langle   K_1, d_1, T_1\rangle_{n_1}$ by using the rule $cl'_r$.
Then, by definition of unfolding there exists a derivation step $\langle   K, d, T\rangle_n\rightarrow_{\omega'_t}
\langle   K_2, d_2, T_2\rangle_{n_2}$ in $P$, by using the rule $cl_r$ and then we have a contradiction.
\end{description}
\end{proof}

\subsection{Safe replacement}

We can now provide the result which shows the correctness of the safe rule replacement condition.
This is done by using  the following proposition.

\begin{proposition}\label{lemma:servcomplete}
Let $cl_r, \, cl_v$ be two annotated CHR rules such that the following holds
\begin{itemize}
\item $cl_r$ is of the form $r@H'_1\backslash H'_2 \Leftrightarrow  D_r\,|\,  K_r; T_r $,
  \item $cl'_r \in Unf_{\{cl_v\}}(cl)$ is of the form $r@H'_1\backslash H'_2 \Leftrightarrow  D_r'\,|\,  K_r'; T'_r $, with
  $\mathcal{CT}  \models D_r \leftrightarrow D'_r$ and it is obtained by unfolding the identified atoms $A \subseteq K_r$.
\end{itemize}
Moreover, let $\sigma$ be a generic built-in free configuration such that we can
construct a derivation $\delta$ from $\sigma$ where
\begin{itemize}
  \item $\delta$ uses at the most the rules $cl_r$ and $cl_v$ in the order,
  \item a built-in free configuration $\sigma^f$  can be obtained and
  \item if $cl_v$ is used, then $cl_v$ rewrites the atoms $A'$ such that $chr(A)=chr(A')$. \comment{corresponding (in the obvious sense) to the considered $A$ in the body of $cl_r$.}
\end{itemize}
Then, we can use the transition \textbf{Apply'} with the rule $cl'_r$ obtaining the configuration $\sigma_{r'}$ and then the built-in free configuration $\sigma_{r'}^f$ such that $\sigma_{r'}^f \simeq \sigma^f$.
\end{proposition}
\begin{proof}
Assume that
$$
\begin{array}{rl}
\sigma&\rrarrow^ {\, cl_r}
\sigma_r \rrarrow^{Solve^{*}}
\sigma_r^f \rrarrow^{cl_v}
\sigma_v\rrarrow^{Solve^{*}}
\sigma_v^f\\
&\searrow_{\, cl'_r}\sigma_{r'}\rrarrow^{Solve^{*}}
\sigma_{r'}^f\\
\end{array}
$$

The labeled arrow $\rrarrow^{Solve^{*}}$ means that only {\bf Solve} transition steps are applied.
Moreover
\begin{itemize}
           \item if $\sigma_r^f$ has the form
$\la   G, {\tt false}, T\ra$ then the derivation between the parenthesis is not present and $\sigma^f=\sigma_r^f$.
           \item the derivation between the parenthesis is present and $\sigma^f=\sigma_v^f$, otherwise.
         \end{itemize}

We first need some notation.  Let  $\sigma=\la (F_1,F_2,F_3) , C, T\ra_j$ be a built-in free configuration
and let $cl_r$ and $cl_v$ be of the form
$r@H_1\backslash H_2 \Leftrightarrow  D_r\,|\,  K,  A, C_r; T_r $ and
$v@H'_1\backslash H'_2 \Leftrightarrow  D_v\,|\,   P, C_v;T_v$ respectively,
$A=A_1 \uplus A_2$, $C_r$ is the conjunction of all the built-in constraints in the body
of $cl_r$ and $\theta$ is a substitution such that $dom(\theta) \subseteq Fv(H'_1, H'_2)$ and
\begin{equation}\label{113marzo1}
    \mathcal{CT}  \models (D_r \wedge C_r) \rightarrow chr(A_1, A_2)=(H'_1, H'_2)\theta .
\end{equation}
Furthermore let $m$ be the greatest identifier which appears in the
rule $cl_r$ and let $(P_1, T_1, m_1)=inst(P,T_v,m)$.

Then, the \emph{unfolded} rule $cl'_r$ is:
\[
r@ H_1\backslash H_2
\Leftrightarrow D_r, (D_v'\theta)\, |\, K, A_1,
P_1, C_r, C_v, chr(A_1, A_2)=(H'_1, H'_2); T'_{r}
\]
where $v @id (A_1, A_2) \not \in T_r$,
$V=\{d\in D_v\mid \mathcal{CT} \models (D_r\wedge C_r)\rightarrow d\theta\}$,
$D'_v= D_v\backslash V$, $Fv(D'_v\theta) \cap Fv(k')\theta \subseteq Fv(H_1, H_2)$, the
constraint $(D_r, (D'_v\theta))$ is satisfiable and
\begin{itemize}
    \item if $H_2'=\epsilon$ then $T'_{r}= T_r \cup T_1 \cup\{v @id (A_1)\}$
    \item if $H_2'\neq \epsilon$ then $T'_{r}=clean((  K,  A_1) , T_r) \cup T_1$.
    \end{itemize}

Since by hypothesis, $\mathcal{CT}  \models (D_r, (D_v'\theta))\leftrightarrow D_r$, we have that
\begin{equation}\label{110dic2}
    \mathcal{CT}  \models (D_r\wedge C_r)\rightarrow D_v\theta \mbox{ and } D_v'\theta=\emptyset.
\end{equation}

Let us now consider the application of the rule
$cl_r$ to $\sigma$. By definition of the \textbf{Apply'} transition step, we have that
\begin{equation}\label{110dic1}
    \mathcal{CT} \models C\rightarrow \exists_{cl_r}((chr(F_1,F_2)=(H_1, H_2))\wedge D_r)
\end{equation}
and
$$\begin{array}{l}
    \sigma_{r}=\la (Q_2, C_r),
          chr(F_1,F_2)= (H_1, H_2)\wedge D_r \wedge C, T_4\ra_{j+m},
  \end{array}
  $$
where $Q_2=(F_1,F_3,K',A')$,
$((K',A'), T_2, j+m) = inst ((K,A), T_r, j)$ and
$T_4=T \cup T_2 \cup\{r @id (F_1, F_2)\}$.

Therefore, by definition
$$ \sigma_{r}^f=\la  Q_2, C_{r}^f, \, T_4\ra_{j+m}.
  $$
  where
\begin{equation}\label{13marzo21}
\mathcal{CT}  \models C_{r}^f \leftrightarrow  C_r\wedge
     chr(F_1,F_2)= (H_1, H_2)\wedge D_r \wedge C.
\end{equation}
Let us now apply the rule $cl'_{r}$ to
$\sigma$. By (\ref{110dic1}), (\ref{110dic2}) and
by definition of the  \textbf{Apply'} transition step, we have that
$$\begin{array}{l}
    \sigma_{r'}=\la (Q, C_r, C_v, chr(A_1, A_2)= (H'_1, H'_2)),D
, T_3\ra_{j+m_1},
  \end{array}
  $$
where
\begin{itemize}
  \item $\mathcal{CT}  \models D \leftrightarrow chr(F_1,F_2)= (H_1, H_2)\wedge D_r \wedge C$,
  \item $Q=(F_1, F_3,  Q_1)$,
  \item $inst ((K, A_1,  P_1), T'_{r}, j)=(Q_1, T''_{r}, j+m_1)$ and
$T_3=T \cup T''_{r} \cup\{r @id (F_1, F_2)\}$.
\end{itemize}

Therefore, by definition
$$\sigma_{r'}^f=\la   Q, C_{r'}^f, \, T_3\ra_{j+m_1}.$$
where
$$\begin{array}{l}
       \mathcal{CT}  \models C_{r'}^f \leftrightarrow  C_r\wedge
       C_v\wedge chr(A_1, A_2)= (H'_1, H'_2)
       \wedge D.
\end{array}$$

Now, we consider the two previously obtained configurations $\sigma_r^f$ and
$\sigma^f_{r'}$. Since by hypothesis $\sigma_v^f$ is a non-failed configuration, we have that $C_{r}^f    \neq \tt false$

Now, let $A' \in Q_2$ such that $chr(A')=chr(A)$. Note that such atoms there exist, since by construction $A$ are atoms in the body of $cl_r$.

  By definition, since
$A$ are atoms in the body of $cl_r$, $dom(\theta) \subseteq Fv(H'_1, H'_2)\subseteq Fv(cl_v)$, by (\ref{13marzo21}), (\ref{113marzo1}) and (\ref{110dic2}), \comment{and (\ref{110dic1}),} we have that
\[\begin{array}{ll}
    \mathcal{CT} \models  & C_{r}^f
     \rightarrow ((chr(A_1, A_2)=(H'_1, H'_2))\wedge D_v)\theta
  \end{array}
 \]
 and therefore, since $dom(\theta) \subseteq Fv(cl_v)$, we have that
\[\begin{array}{ll}
    \mathcal{CT} \models  & C_{r}^f
     \rightarrow \exists_{cl_v}((chr(A_1, A_2)=(H'_1, H'_2))\wedge D_v).
  \end{array}
 \]
Then, since by hypothesis $cl_v$ rewrites the atom $A=(A_1, A_2)$ such that $chr(A')=chr(A'_1, A'_2)=chr (A_1, A_2)=chr(A)$, we have that
$$\begin{array}{ll}
    \sigma_{v}= \la (Q_3 ,  C_v), D', T_5\ra_{m_1},
  \end{array}
  $$
     where
     \begin{itemize}
       \item $ Q_3=(F_1, F_3,  K',  A'_1,   P_2)$,
       \item  $D'= (chr(A_1, A_2)=(H'_1, H'_2)\wedge D_v \wedge C_r\wedge
     chr(F_1,F_2)= (H_1, H_2)\wedge D_r \wedge C)$,
       \item $inst (  P, T_v, j+m) = (  P_2, T_v', m_1)$ and
       $T_5=T_4 \cup T_v' \cup\{v @id (
k'')\}$.
     \end{itemize}
Finally, by definition, we have that
$\sigma_{v}^f=\la   Q_3,C_{v}^f, \, T_5\ra_{m_1},$ where
  $$\begin{array}{ll}
    \mathcal{CT}  \models & C_{v}^f \leftrightarrow (C_v \wedge D').
  \end{array}
  $$

  If $C_{v}^f =\tt false $ then the proof is analogous to the previous case and hence it is omitted.

  Otherwise, the proof is analogous to that given for Proposition~\ref{prop:servequality} and hence it is omitted.
\end{proof}

\begin{proposition}\label{propcambioequality}
Let $\sigma_0=\langle F,c, T \rangle_m $ be a built-in configuration and let $cl$ be an annotated CHR rule such that the following holds.\begin{description}
        \item[a)] $cl= r@H_1\backslash H_2 \Leftrightarrow D\,|\,  A; T$. where
        $(H_1,  H_2)=(h_1, \dots, h_n)$,
        \item[b)]
        there exists $(K_1, K_2)= (k_1, \dots, k_n) \subseteq F$ such that
        ${\it r}@id (K_1,K_2) \not \in T $ and $\mathcal{CT}  \models c \rightarrow \exists _{cl}
        ((chr(  K_1,   K_2)=(H_1,  H_2))\wedge D)$,
        \item[c)] there exist $l \in \{1, \dots, n\}$ and $(K'_1, K'_2)= (k'_1, \dots, k'_n) \subseteq F$ such that
        $k_l=k'_l$,
        ${\it r}@id (K'_1,K'_2) \not \in T $ and $\mathcal{CT}  \models c \wedge chr(k_l)=h_l \rightarrow
        (chr(  K'_1,   K'_2)=(H_1,  H_2))$,
        \item[d)] $\sigma_0\rrarrow _{\omega'_t}\sigma$ is an $\textbf{Apply'}$ transition step which uses the clause $cl$, rewrites the atoms $(K_1, K_2)$ and such that
            $\sigma=\langle ((F \setminus K_1) \uplus A'),C, T ' \rangle_{m'}$, where $C$ is the constraint
            $(chr(  K_1,   K_2)=(H_1,  H_2)) \wedge D \wedge c$.
      \end{description}
      Then, there exists an $\textbf{Apply'}$  transition step  $\sigma_0\rrarrow _{\omega'_t}\sigma'$  which uses the clause $cl$, rewrites the atoms $(K'_1, K'_2)$ and such that
      $\sigma'=\langle ((F \setminus K'_1) \uplus A'),C', T "\rangle_{m'}$, where $C'$ is the constraint
            $(chr(  K'_1,   K'_2)=(H_1,  H_2)) \wedge D \wedge c$ and
            \begin{enumerate}
              \item \label{28sett1} $\mathcal{CT}  \models ((chr (F \setminus K_1) \wedge  A') \wedge C) \leftrightarrow ((chr (F \setminus K'_1) \wedge  A') \wedge C') $,
              \item \label{28sett2} $T''=(T' \setminus\{{\it r}@id (K_1,K_2)\}) \cup {\it r}@id (K'_1,K'_2)$.
            \end{enumerate}
      \end{proposition}
\begin{proof}
First of all, by definition of $\textbf{Apply'}$  transition step and since, by hypothesis ${\bf c)}$, ${\it r}@id (K'_1,K'_2) \not \in T $, we have to prove that
\[\mathcal{CT}  \models c \rightarrow \exists _{cl}
        ((chr(  K'_1,   K'_2)=(H_1,  H_2))\wedge D).
        \]
         By hypothesis $
        {\bf b)}$ and since $Fv(c) \cap Fv(cl)= \emptyset$, we have that
         $\mathcal{CT}  \models c \rightarrow \exists _{cl} (c \wedge
        (chr( k_l)=h_l)\wedge D)$. Hence the thesis follows by hypothesis ${\bf c)}$.

        Now, we have to prove~\ref{28sett1}. By hypothesis ${\bf b)}$,  we have that
        $\mathcal{CT}  \models c  \rightarrow \exists _{cl}(chr(  K_1,   K_2)=(H_1,  H_2))$.
        Therefore, there exists a substitution $\vartheta$ such that $dom(\vartheta)= Fv(H_1,  H_2)$ and
\begin{equation}\label{7sett111}
            \mathcal{CT}  \models c  \rightarrow (chr(  K_1,   K_2)=(H_1,  H_2)\vartheta).
        \end{equation}
        By hypothesis ${\bf c)}$ and since $dom(\vartheta)  \cap Fv(c, K_1, K_2) = \emptyset$, we have that
        \[\mathcal{CT}  \models (c \wedge (chr(k_l)=h_l  \vartheta ))\rightarrow
        (chr(  K'_1,   K'_2)=(H_1,  H_2)\vartheta)
        \]
         and by (\ref{7sett111}),
        $\mathcal{CT}  \models c  \rightarrow (chr( k_l)=(h_l)\vartheta)$.\\
        Then $\mathcal{CT}  \models c  \rightarrow (chr(  K'_1,   K'_2)=chr(  K_1,   K_2))$  and then the thesis.
        \\
        The proof of~\ref{28sett2} is obvious by definition of $\textbf{Apply'}$  transition step.
\end{proof}

\begin{proposition}\label{propavanza}
Let $\sigma_0=\langle F,c, T \rangle_m $ be a built-in configuration such that there exists a normal terminating derivation $\delta$ starting from $\sigma$ which ends in a configuration $\sigma$. Assume that $\delta$ uses an annotated CHR rule $cl$ such that the following holds.
\begin{description}
        \item[a)] $cl= r@H_1\backslash H_2 \Leftrightarrow D\,|\,  A; T$
        \item[b)]
        there exists $(K_1, K_2)\subseteq F$ such that $cl$ rewrites the atoms $(K_1, K_2)$ in $\delta$ and $\mathcal{CT}  \models c \rightarrow \exists _{cl} ((chr(K_1, K_2)=(H_1,  H_2))\wedge D)$

      \end{description}
Then, there exists a normal terminating derivation $\delta'$ starting from $\sigma_0$ such that
\begin{itemize}
  \item $\delta'$ uses at most the same clauses of $\delta$ and uses the rule $cl$ in the first $\bf Apply'$ transition step, in order to rewrite the atoms $(K_1, K_2)$,
  \item $\delta'$ ends in a configuration $\sigma' $ such that $\sigma \simeq \sigma'$.
\end{itemize}
\end{proposition}
\begin{proof}
The proof is obvious by definition of derivation.
\end{proof}

\noindent Hence we have the following result.\\
\setcounter{theorem}{0}
\begin{theorem}
Let $P$ be an annotated program,  $cl$ be a rule in $P$ such that
$cl$ can be safely replaced
in $P$ according to Definition~\ref{def:nsafedel}. Assume also that
\[\begin{array}{l}
  P'  =  (P\setminus   \{cl\} ) \, \cup
   Unf_{P}(cl).
\end{array}
\]

Then $\mathcal{QA'}_{P}(G)=\mathcal{QA'}_{P'}(G)$ for any arbitrary
goal $G$.
\end{theorem}
\setcounter{theorem}{3}

\begin{proof}
By using a straightforward inductive argument and by
Proposition~\ref{prop:equality}, we have that
$\mathcal{QA'}_{P}(G)=\mathcal{QA'}_{P''}(G)$ where
 \[
  P'' = P \, \cup \, Unf_{P}(cl),
\]
for any arbitrary goal $G$.

Then, to prove the thesis, we have only to prove that
\[\mathcal{QA'}_{P'}(G) = \mathcal{QA'}_{P''}(G).\]
In the following, we assume that $cl$ is of the form $r@H_1\backslash H_2 \Leftrightarrow D\,|\,  A; T$.
We prove the two inclusions separately.
\begin{description}
  \item[{\bf ($\mathcal{QA'}_{P'}(G) \subseteq \mathcal{QA'}_{P''}(G)$) }]
  The proof is by contradiction. Assume that there exists $Q \in \mathcal{QA'}_{P'}(G) \setminus \mathcal{QA'}_{P''}(G)$. By definition there exists a derivation
  \[\delta=\langle I(G),
\texttt{true}, \emptyset\rangle_m\rightarrow^{*}_{\omega'_t}
\langle   K, d, T\rangle_n\not\rightarrow_{\omega'_t}\]
in $P'$, such that $Q =
\exists _{-Fv(G)}(chr(   K)\wedge d)$. Since $P' \subseteq P''$, we have that there exists the derivation
\[\langle I (G),
\texttt{true}, \emptyset\rangle_m\rightarrow^{*}_{\omega'_t}
\langle   K, d, T\rangle_n\] in $P''$. Moreover, since $P'' =P '\cup \{cl\} $ and $Q\not  \in \mathcal{QA'}_{P''}(G)$, we have that there exists a derivation step $\langle   K, d, T\rangle_n\rightarrow_{\omega'_t}
\langle   K_1, d_1, T_1\rangle_{n_1}$ by using the rule $cl$. \\
Since $cl$ can be safely replaced
in $P$, we have that there exists an unfolded rule $cl' \in Unf_{P}(cl)$ such that $cl'$ is of the form \[r@H_1\backslash H_2 \Leftrightarrow  D'\,|\,  A'; T',\]
$\mathcal{CT}  \models D \leftrightarrow D'$ and by construction $cl' \in P'$. \\
Then, there exists a derivation step $\langle   K, d, T\rangle_n\rightarrow_{\omega'_t}
\langle   K_2, d_2, T_2\rangle_{n_2}$ in $P'$ (by using the rule $cl'$) and then we have a contradiction.

  \item[{\bf ($\mathcal{QA'}_{P''}(G) \subseteq \mathcal{QA'}_{P'}(G)$) }]
  First of all, observe that by Proposition~\ref{prop:solonorm},
  $\mathcal{QA'}_{P''}(G)$ can be calculated by
considering only non-failed normal terminating derivations. \\
Then, for each non-failed normal terminating derivation $\delta$ in $P''$, which uses the rule $cl$ after the application of $cl$, we obtain the configuration $\sigma_{1}$ and then a non-failed built-in free configuration $\sigma_{1}^f$. Now, let $C$ be the built-in constraint store of $\sigma_1^f$.

Since by hypothesis
$cl$ can be safely replaced
in $P$, following Definition~\ref{def:nsafedel}, we have there exists at least an atom
$ k \in   A$, such that there exists a corresponding atom (in the obvious sense) $k'$ which is rewritten in $\delta$ by using a rule $cl'$ in $P$. Therefore, without loss of generality, we can assume that

\[\delta=\langle I (G),
\texttt{true}, \emptyset\rangle_m\rightarrow^{*}_{\omega'_t} \sigma \rightarrow_{\omega'_t} \sigma_{1}\rightarrow^{*}_{\omega'_t}\sigma_{1}^f\rightarrow^{*}_{\omega'_t} \sigma_2 \rightarrow_{\omega'_t} \sigma_{3}
\rightarrow^{*}_{\omega'_t}\sigma_{4}
\not \rightarrow^{*}_{\omega'_t}\]
 where the transition step $s_1=\sigma \rightarrow_{\omega'_t} \sigma_{1}$ is the first {\bf Apply'} transition step which uses the clause $cl$ and $ s_2= \sigma_2 \rightarrow_{\omega'_t} \sigma_{3}$ is the first {\bf Apply'} transition step which rewrites an atom $k'$, corresponding to an atom $k$ in the body of $cl$ introduced by $s_1$.
 Since by hypothesis
$cl$ can be safely replaced
in $P$ and by Proposition~\ref{propcambioequality} we can assume that $cl'$ rewrites in $s_2$ only atoms corresponding (in the obvious sense) to atoms in $A$. Moreover, since by hypothesis
$cl$ can be safely replaced
in $P$ and by Proposition~\ref{propavanza}, we can assume that $s_2$ is the first {\bf Apply'} transition step after $s_1$.
Then, the thesis follows since by hypothesis
$cl$ can be safely replaced
in $P$, by Proposition~\ref{lemma:servcomplete} and by a straightforward inductive argument.

\end{description}
\end{proof}

\subsection{Termination and confluence}

We first prove the  correctness of unfolding w.r.t. termination.\\

\setcounter{proposition}{2}
\begin{proposition}[{\sc Normal Termination}]

Let $P$ be a CHR program and
let $P_0, \ldots, P_n$ be an U-sequence starting from $Ann(P)$. $P$ satisfies
normal termination if and only if $P_n$  satisfies normal termination.
\end{proposition}
\setcounter{proposition}{7}

\begin{proof}
By Lemma~\ref{lemma:intermequiv}, we have that $P$ is normally terminating if and only if $Ann(P)$ is normally terminating.
Moreover from Proposition~\ref{prop:servequality} and Proposition~\ref{lemma:servcomplete} and by using a straightforward inductive argument, we have that
for each $i=0, \ldots,n-1$, $P_i$ satisfies normal termination if and only if $P_{i+1}$ satisfies the normal termination too and then the thesis.
\end{proof}

The following lemma relates the $\approx$, $\simeq$
and $\equiv_V$ equivalences.

\begin{lemma}\label{lem:relrel}
Let $\sigma, \sigma'$ be final configurations in ${\it Conf_t}$, $\sigma_1, \sigma_2, \sigma'_1, \sigma'_2 \in {\it Conf'_t}$ and let $V$ be a set of variables.
\begin{itemize}
  \item If $\sigma_1\approx \sigma$, $\sigma'_1\approx \sigma'$ then  $\sigma_1\equiv_{V}\sigma'_1$ if and only if $\sigma\equiv_{V}\,\sigma'$.
  \item If $\sigma_1\simeq \sigma_2$, $\sigma'_1\simeq \sigma'_2$ and  $\sigma_1\equiv_{V}\,\sigma'_1$ then $\sigma_2\equiv_{V}\,\sigma'_2$.
\end{itemize}
\end{lemma}
\begin{proof}
The proof of the first statement follows by definition of $\approx$ and by observing that if $\sigma$ is a final configuration in ${\it Conf_t}$, then $\sigma$
has the form $\langle  G,
  S, {\tt false} , T\rangle_n$  or it
has the form $\langle \emptyset,  S, c,T \rangle_n$.

The proof of the second statement is straightforward, by observing that if $\sigma_1\simeq \sigma_2$, then
$\sigma_1 \equiv_{V}\,\sigma_2$ for each set of variables $V$.
\end{proof}

\setcounter{theorem}{1}
\begin{theorem}[{\sc Confluence}] Let $P$ be a CHR program and let $P_0, \ldots, P_n$
be an NRU-sequence starting from $P_0= Ann(P)$. $P$ satisfies confluence if and only if $P_n$
satisfies confluence too.
\end{theorem}\setcounter{theorem}{3}
\begin{proof}
\comment{We prove only that if $P$ is normally terminating and confluent, then $P_n$ is confluent too.
The proof of the converse is similar and hence it is omitted.}
By Lemma~\ref{lemma:intermequiv}, we have that $P$ is confluent if and only if $Ann(P)$ is confluent.
Moreover, by Proposition~\ref{prop:servequality}
Now, we prove that
for each $i=0, \ldots,n-1$, $P_i$ is confluent if and only if $P_{i+1} = (P_i \setminus   \{cl^i \} ) \, \cup  \,Unf_{P_i}(cl^i)$ is confluent too.
Then, the proof follows by a straightforward inductive argument.
\begin{itemize}

  \item Assume that $P_i$ is confluent and let us assume by contrary that $P_{i+1}$ does not satisfies confluence.
  By definition,  there exists a state $\sigma = \langle (  K,D),C,
T\rangle_o$ and two derivations $\sigma\rrarrow_{\omega'_t}^{*} \sigma_1$ and $\sigma\rrarrow_{\omega'_t}^{*}\sigma_2$ in $P_{i+1}$ such that
there are no two derivations $\sigma_1 \mapsto^{*}\sigma_1'$ and $\sigma_2\mapsto^{*}\sigma_2'$ in  $P_{i+1}$ where
$\sigma_1'\equiv_{Fv(\sigma)}\,\sigma_2'$. Without loss of generality, we can assume that $\sigma_1$ and $\sigma_2$ are built-in free states.  Therefore, by Proposition~\ref{prop:servequality}, there exist
two derivations $\sigma\rrarrow_{\omega'_t}^{*} \sigma_3$ and $\sigma\rrarrow_{\omega'_t}^{*}\sigma_4$ in $P_{i}$, such that $\sigma_1\simeq \sigma_3$ and $\sigma_2\simeq \sigma_4$. Moreover, since $P_i$ is confluent, there exist two derivations $\delta=\sigma_3\rrarrow_{\omega'_t}^{*} \sigma'_3$ and $\delta'=\sigma_4\rrarrow_{\omega'_t}^{*}\sigma'_4$ in $P_{i}$ such that $\sigma'_3 \equiv_{Fv(\sigma)} \sigma'_4$. Moreover, without loss of generality, we can assume that $\sigma'_3$ and $\sigma'_4$ are built-in free.
Analogously to Theorem~\ref{theo:n1completeness}, since by hypothesis
$cl^i$ can be safely replaced
in $P_i$ and by using Proposition~\ref{lemma:servcomplete}, we can counstruct two new derivations
$\gamma=\sigma_1\rrarrow_{\omega'_t}^{*} \sigma_5$ and $\gamma'=\sigma_2\rrarrow_{\omega'_t}^{*}\sigma_6$
in $P_{i} \, \cup  \,Unf_{P_i}(cl^i)$ such that
$\sigma_5$ and $\sigma_6$ are built-in free,
$\sigma_5 \simeq  \sigma'_3$, $\sigma_6 \simeq  \sigma'_4$ and such that if $\gamma$ and $\gamma'$ use the clause $cl^i$, then no atoms introduced (in the obvious sense) by $cl_i$ is rewritten by using (at least) one rule in $P_{i} \, \cup  \,Unf_{P_i}(cl^i)$.
Moreover, by hypothesis and by Lemma~\ref{lem:relrel}, $\sigma_5 \equiv_{Fv(\sigma)} \sigma_6$.

Let $l$ the number of the {\bf Apply'} transition steps in $\delta$ and $\delta'$,
which use the rule $cl^i$ and whose body is not rewritten by using (at least) one rule in $P_{i} \, \cup  \,Unf_{P_i}(cl^i)$.
The proof is by induction on $l$.
\begin{description}
  \item[($l=0$)] In this case, $\gamma=\sigma_1\rrarrow_{\omega'_t}^{*} \sigma_5$ and $\gamma'=\sigma_2\rrarrow_{\omega'_t}^{*}\sigma_6$ are derivations in $P_{i+1}$. By hypothesis  and by Lemma~\ref{lem:relrel}, we have that  $\sigma_5 \equiv_{Fv(\sigma)} \sigma_6$ and then we have a contradiction.

\item[($l>0$)] Let us consider the last {\bf Apply'} transition step in $\gamma$ and $\gamma'$, which use (a renamed version of) the rule $cl^i=r_i@H_1\backslash H_2 \Leftrightarrow D\,|\,  K, C; T$, whose body is not rewritten by using (at least) one rule in $P_i \, \cup  \,Unf_{P_i}(cl^i)$ and where $C$ is the conjunction of all the built-in constraints in the body of $cl^i$. Without loss of generality, we can assume that such an {\bf Apply'} transition step is in $\gamma$.
    Now, we have two possibilities

    \begin{itemize}
      \item $\sigma_5$ is a failed configuration. By definition of $\equiv_{Fv(\sigma)}$, we have that $\sigma$ is also a failed configuration. In this case, it is easy to check that, by using Lemma~\ref{lemma:servcomplete}, we can substitute each {\bf Apply'} transition steps in $\delta$ and $\delta'$, which use the rule $cl^i$ and whose body is not rewritten by using (at least) one rule $P_i$, with an {\bf Apply'} transition step which uses a rule in $Unf_{P_i}(cl^i)\subseteq P_{i+1}$. Then, analogously to the case {\bf ($l=0$)}, it is easy to check that there exist the derivations $\gamma_1=\sigma_1\rrarrow_{\omega'_t}^{*} \sigma'_5$ and $\gamma_1'=\sigma_2\rrarrow_{\omega'_t}^{*}\sigma'_6$ in $P_{i+1}$ such that $ \sigma^f_3$ and $\sigma^f_4$ are both failed configurations and then we have a contradiction.

\item  $\sigma_5$ is not a failed configuration. Then $\sigma_5$ is of the form $\langle S_5,C_5, T_5\rangle_{n_5}$, where $chr(K) \subseteq chr(S_5)$.
    Moreover, since $cl^i$ can be non-recursively safely replaced in $P_i$, there exists a clause $cl_v$ in  $P_i \setminus \{cl^i\}$ such that
    $cl^i$ can be unfolded by using $cl_v$. Therefore, by definition of non-recursive safe unfolding,
    there exists a new derivation $\gamma_1=\sigma_1\rrarrow_{\omega'_t}^{*} \sigma_5\rrarrow_{\omega'_t}^{*}\sigma'_5$, where $\sigma'_5$ is obtained from $\sigma_5$  first by an
     {\bf Apply'} transition step, which uses the rule $cl_v$ and rewrites atoms in the body of $cl^i$ and then some {\bf Solve'} transition steps.
     By definition of $\equiv_{Fv(\sigma)}$ and since $\sigma_5 \equiv_{Fv(\sigma)} \sigma_6$, we have that there exists also a new derivation
     $\gamma_1'=\sigma_2\rrarrow_{\omega'_t}^{*}\sigma_6\rrarrow_{\omega'_t}^{*}\sigma'_6$, where $\sigma'_5$ is obtained from $\sigma_5$  first by an
     {\bf Apply'} transition step, which uses the rule $cl_v$ and rewrites atoms in the body of $cl^i$ and then some {\bf Solve'} transition steps.

 Since by hypothesis $\sigma_5 \equiv_{Fv(\sigma)} \sigma_6$, we have that $\sigma'_5 \equiv_{Fv(\sigma)} \sigma'_6$. Moreover the number of the {\bf Apply'} transition steps in $\delta_2$ and $\delta'_2$, which use the rule $cl^i$ whose body is not rewritten by using (at least) one rule in $P_i$
     is strictly less than $l$ and then the thesis.
     \end{itemize}
   \end{description}
    \item Assume that $P_{i+1}$ is confluent and let us assume by contrary that $P_i$  does not satisfies confluence. The proof is analogous to the previous case and hence it is omitted.
\end{itemize}
\end{proof}

\subsection{Weak safe rule replacement}

Finally, we provide the proof of Proposition~\ref{prop:wterm}. We first need of the following lemma, which provides an alternative characterization of confluence for normally terminating programs.

\begin{lemma}\label{conflnormterm}
Let $P$ be a CHR [annotated] normally terminating program. $P$ is confluent if and only if for each pair of normal derivations $\sigma\mapsto^{*} \sigma_1^f\not \mapsto^{*} $ and $\sigma\mapsto^{*} \sigma_2^f\not\mapsto^{*} $, we have that
$\sigma_1^f \equiv_{Fv(\sigma)}\sigma_2^f$.
\end{lemma}
\begin{proof}
\begin{description}
  \item[(Only if)] The proof is straightforward by definition of confluence.
  \item[(If)] The proof is by contradiction. Assume that $P$ is not confluent. Then, there exists a state $\sigma$ such that $\sigma\mapsto^{*} \sigma_1$ and $\sigma\mapsto^{*} \sigma_2$ and for each pair of states $\sigma_f'$
and $\sigma_f''$ such that $\sigma_1 \mapsto^{*}\sigma_f'$ and $\sigma_2\mapsto^{*}\sigma_f''$, we have that $\sigma_f'\not \equiv_{Fv(\sigma)}\,\sigma_f''$. In particular, since $P$ is normally terminating, we have that there exists $\sigma_f'$ and $\sigma_f''$ such that
$\sigma_1 \mapsto^{*}\sigma_f'\not \mapsto^{*}$, $\sigma_2\mapsto^{*}\sigma_f''\not \mapsto^{*}$ and $\sigma_f'\not \equiv_{Fv(\sigma)}\,\sigma_f''$.
Then, it is easy to check that there exist two normal derivation
$\sigma \mapsto^{*}\sigma_1'\not \mapsto^{*}$ and $\sigma\mapsto^{*}\sigma_2'\not \mapsto^{*}$ such that
$\sigma_f'\simeq \sigma_1'$ and $\sigma_f''\simeq \sigma'_2$. Since $\sigma_f'\not \equiv_{Fv(\sigma)}\,\sigma_f''$, by definition of $\simeq$, we have that $\sigma_1'\not \equiv_{Fv(\sigma)}\,\sigma_2'$ and then we have a contradiction.
\end{description}
\end{proof}

Then, we have the desired result.

\setcounter{proposition}{3}
\begin{proposition}
Let $P$ be an annotated CHR program and let $cl \in P$ such that $cl $
can be weakly safely replaced (by its unfolded version) in $P$. Moreover let
\[P'  =  (P\setminus   \{cl\} ) \, \cup \,
   Unf_{P}(cl).\]
If $P$ is normally terminating then $P'$ is normally terminating. Moreover, if $P$ is normally terminating and confluent then $P'$ is confluent too.
\end{proposition}
\setcounter{proposition}{7}

\begin{proof}
First, we prove that if $P$ is normally terminating then  $P''$ is normally terminating too, where
\[P''= P \, \cup
   Unf_{P}(cl). \]
   Then, we prove that if $P''$ is normally terminating then $P'$ is normally terminating.
    Analogously if $P$ is normally terminating and confluent and then the thesis.
\begin{itemize}
   \item Assume that $P$ is normally terminating.
  The proof of the normal termination of $P''$ follows by Proposition~\ref{prop:servequality}.

   \item Now, assume that $P$ is normally terminating and confluent and by the contrary that $P''$ does not satisfy confluence.

       By Lemma~\ref{conflnormterm} and since by the previous result $P''$ is normally terminating, there exist a state
$\sigma$ and two normal derivations
\[\sigma\rrarrow_{\omega'_t}^{*}\sigma'_f \not \rrarrow_{\omega'_t} \mbox{ and } \sigma\rrarrow_{\omega'_t}^{*}\sigma''_f\not \rrarrow_{\omega'_t}
\] in $P''$
such that
$\sigma'_f \not \equiv_{Fv(\sigma)}\sigma''_f$.

Then, by using arguments similar to that given in Proposition~\ref{prop:servequality} and since $P\subseteq P''$, we have that there exist two normal derivations
\[\sigma\rrarrow_{\omega'_t}^{*}\sigma_1^f \not \rrarrow_{\omega'_t}\mbox{ and } \sigma'\rrarrow_{\omega'_t}^{*}\sigma_2^f
\not \rrarrow_{\omega'_t}\]
in $P$, where
 $\sigma'_f \simeq \sigma_1^f$ and $\sigma_f'' \simeq \sigma_2^f$. Since by hypothesis $P$ is confluent, we have that $\sigma_1^f \equiv_{Fv(\sigma)}\sigma_2^f$.
  Therefore, by Lemma~\ref{lem:relrel} we have a contradiction to the assumption that there exist two states $\sigma'_f$ and $\sigma''_f$ as previously defined.
 \end{itemize}

 \noindent Now, we prove that if $P''$ is normally terminating then $P'$ is normally terminating. Moreover we prove  that if $P''$ is normally terminating and confluent then $P'$ is confluent too and then the thesis.
 \begin{itemize}
   \item If $P''$ is normally terminating then, since $P'\subseteq P''$, we have that
  $P'$ is normally terminating too.
     \item Now, assume that $P''$ is normally terminating and confluent and  by the contrary that $P'$ does not satisfy confluence. Moreover, assume that $cl$ is of the form $r@H_1\backslash H_2 \Leftrightarrow D\,|\,  A; T$. By Lemma~\ref{conflnormterm} and since by the previous result $P'$ is normally terminating, there exist a state
$\sigma$ and two normal derivations
\[\sigma\rrarrow_{\omega'_t}^{*}\sigma_1^f\not \rrarrow_{\omega'_t} \mbox{ and } \sigma\rrarrow_{\omega'_t}^{*}\sigma_2^f\not \rrarrow_{\omega'_t}
\]
 in $P'$ such that
$\sigma_1^f \not \equiv_{Fv(\sigma)}\sigma_2^f$.

Since $P' \subseteq P''$, we have that there exist two normal derivations
\[\sigma\rrarrow_{\omega_t}^{*}\sigma_1^f \mbox{ and } \sigma'\rrarrow_{\omega_t}^{*}\sigma_2^f\]
in $P''$. Then, since $P''$ is confluent and $P''= P '\cup \{cl\}$ there exists $i \in [1,2]$ such that
$\sigma_i^f\rrarrow_{\omega'_t}\sigma'$ in $P''$ by using the rule $cl \in (P'' \setminus P')$. In this case, by definition of weak safe replacement, there exists an unfolded rule $cl' \in Unf_P (cl)$  such that $cl'$ is of the form
$$r@ H_1\backslash H_2 \Leftrightarrow D'\, |\,   A'; T'
   $$
with $\mathcal{CT}  \models D \leftrightarrow D'$ and by construction $cl' \in P'$. Therefore $\sigma_i^f\rrarrow_{\omega'_t}\sigma''$ in $P'$, by using the rule $cl'$, and then we have a contradiction.
 \end{itemize}
\end{proof}

\end{document}